\def\version{November 25, 2014}

\documentclass[12pt]{article}
\def\macrosPb{}
\def\macrosHarxiv{}
\usepackage[psamsfonts]{amsfonts}
\usepackage{amsmath,amssymb,amsthm}
\usepackage[dvips]{graphicx}
\usepackage{appendix}
\usepackage{bbm} 
\usepackage{amsbsy}
\usepackage{enumerate}
\usepackage{cite}

\ifdefined\macrosPa
  \usepackage[textwidth=465pt,textheight=650pt,centering]{geometry} 
\else\ifdefined\macrosPb
  \usepackage[textwidth=500pt,textheight=650pt,centering]{geometry} 
\fi\fi

\ifdefined\macrosS


  \usepackage{mathptmx}
  \DeclareMathAlphabet{\mathcal}{OMS}{cmsy}{m}{n}
\fi

\def\UseSection{
        \numberwithin{equation}{section}
	\theoremstyle{plain}
        \newtheorem{theorem}    {Theorem}[section]
        \DefineTheorems 
}

\def\DefineTheorems{
	
	\newtheorem{lemma}      [theorem] {Lemma}
	
	\newtheorem{prop}       [theorem] {Proposition}
	
	\newtheorem{cor}        [theorem] {Corollary}

	\theoremstyle{definition}
	\newtheorem{defn}       [theorem] {Definition}

	\newtheorem{rk} 	[theorem] {Remark}
	\theoremstyle{definition}

}

\newcommand{\bt}   {\begin{theorem}}
\newcommand{\et}   {\end  {theorem}}
\newcommand{\bl}   {\begin{lemma}}
\newcommand{\el}   {\end  {lemma}}
\newcommand{\bp}   {\begin{prop}}
\newcommand{\ep}   {\end  {prop}}
\newcommand{\bc}   {\begin{cor}}
\newcommand{\ec}   {\end  {cor}}
\newcommand{\bd}   {\begin{defn}}
\newcommand{\ed}   {\end  {defn}}

\newcommand{\ba}   {\begin{array}}
\newcommand{\ea}   {\end  {array}}
\newcommand{\be}   {\begin{enumerate}}
\newcommand{\ee}   {\end  {enumerate}}
\newcommand{\bi}   {\begin{itemize}}
\newcommand{\ei}   {\end  {itemize}}

\def\eq#1\en{\begin{equation}#1\end{equation}}  
\def\eqsplit#1\ensplit{
	\begin{equation}\begin{split}#1\end{split}\end{equation}
	}
\def\eqalign#1\enalign{
	\begin{align}#1\end{align}
	}
\def\eqmul#1\enmul{
	\begin{multline}#1\end{multline}
	}
\newcommand{\eqarrstar} {\begin{eqnarray*}} 
\newcommand{\enarrstar} {\end{eqnarray*}} 
\newcommand{\eqarray}   {\begin{eqnarray}} 
\newcommand{\enarray}   {\end{eqnarray}} 
\newcommand{\nnb}	{\nonumber \\} 

\newcommand{\lbeq}[1]  {\label{e:#1}}
\newcommand{\refeq}[1] {\eqref{e:#1}}    

\makeatletter
\newcommand{\labelcounter}[2]{{%
	\stepcounter{#1}
	\protected@write\@auxout{}%
	{\string\newlabel{#2}{{\csname the#1\endcsname}{\thepage}}}%
	{\ref{#2}}
	}}
\makeatother
\newcommand{\Nbold} {{\mathbb N}}

\newcommand{\Rbold} {{\mathbb R}}

\newcommand{\Zbold} {{\mathbb Z}}

\newcommand{\Bcal}   {\mathcal{B}} 
 
\newcommand{\Dcal}   {\mathcal{D}}

\newcommand{\Ical}   {\mathcal{I}}

\newcommand{\Ncal}   {\mathcal{N}} 
 
\newcommand{\Pcal}   {\mathcal{P}}
\newcommand{\Qcal}   {\mathcal{Q}}

\newcommand{\Wcal}   {\mathcal{W}}

\newcommand{\Zd}    {{ {\Zbold}^d }}

\newcommand{\spose}[1] {{\hbox to 0pt{#1\hss}} }
\newcommand{\ltapprox} {\mathrel{\spose{\lower 3pt\hbox{$\mathchar"218$}}
 \raise 2.0pt\hbox{$\mathchar"13C$}}}
\newcommand{\gtapprox} {\mathrel{\spose{\lower 3pt\hbox{$\mathchar"218$}}
 \raise 2.0pt\hbox{$\mathchar"13E$}}}

\UseSection   
\usepackage[usenames]{color}

\definecolor{at}{rgb}{0.0, 0.5, 0.0} 

\newcommand{\DV}{\Dcal}
\newcommand{\DVa}{\alpha}

\renewcommand{\to} {\rightarrow}

\newcommand{\R}{\Rbold}
\newcommand{\Z}{\Zbold}

\newcommand{\N}{\Nbold}
\newcommand{\C}{\mathbb{C}}

\newcommand{\1}{\mathbbm{1}}

\newcommand{\Ex}{\mathbb{E}}

\newcommand{\chicCov}{{\chi}}

\newcommand{\diam}[1]{\textrm{diam}(#1)}

\newcommand{\pt}{{\rm pt}}

\newcommand{\Vbulk}{U}

\newcommand{\lambdach}{\check{\lambda}}

\newcommand{\vch}{\check{v}}

\newcommand{\gbar}{\bar{g}}

\newcommand{\ggen}{\tilde{g}}
\newcommand{\sgen}{\tilde{s}}
\newcommand{\chigen}{\tilde{\chi}}
\newcommand{\mgen}{\tilde{m}}
\newcommand{\Iint}{\mathbb{I}}
\newcommand{\Igen}{\tilde{\mathbb{I}}}

\newcommand{\domRG}{\mathbb{D}}

\newcommand{\pp}{a}
\newcommand{\qq}{b}
\newcommand{\sigmaa}{\sigma}
\newcommand{\sigmab}{\bar{\sigma}}

\newcommand{\half}{\textstyle{\frac 12}}

\newcommand{\phib}{\bar\phi}

\ifdefined\macrosH
  \usepackage{xr-hyper}
  \usepackage{hyperref}
  \hypersetup{hypertexnames=false}
  \hypersetup{colorlinks,citecolor=blue,linkcolor=blue}  

\else\ifdefined\macrosHarxiv
  \usepackage{xr-hyper}
  \usepackage{hyperref}
  \hypersetup{hypertexnames=false}

\else
  
  \usepackage{xr}
\fi\fi

\title{
  Critical two-point function of the
  \\
  4-dimensional
  weakly self-avoiding walk
}

\author{
  Roland Bauerschmidt\thanks{School of Mathematics,
    Institute for Advanced Study,
    Einstein Drive,
    Princeton, NJ 08540 USA.
    E-mail: {\tt roland@bauerschmidt.ca}.},\;
  David C.\ Brydges\thanks{Department of Mathematics,
    University of British Columbia,
    Vancouver, BC, Canada V6T 1Z2.
    E-mail: {\tt db5d@math.ubc.ca}, {\tt slade@math.ubc.ca}.}\;
  and Gordon Slade$^\dagger$}

\date\version

\begin{document}

\maketitle

\begin{abstract}
We prove $|x|^{-2}$
decay of the critical two-point function
for the continuous-time
weakly self-avoiding walk on $\Zd$, in the
upper critical dimension $d=4$.
This is a statement that the critical exponent $\eta$ exists and
is equal to zero.  Results of this nature have been proved
previously for dimensions $d \ge 5$ using the lace expansion, but the lace
expansion does not apply when $d=4$.
The proof is based on a rigorous renormalisation group analysis of an exact
representation of the continuous-time
weakly self-avoiding walk as
a supersymmetric field theory.
Much of the analysis applies more widely
and has been carried out in a previous paper,
where an asymptotic formula for the susceptibility is
obtained.  Here, we show how observables can be incorporated into the
analysis to obtain a pointwise
asymptotic formula for the critical two-point function.
This involves perturbative calculations similar to those familiar in the physics
literature, but with error terms controlled rigorously.
\end{abstract}

\section{Main result}

\subsection{Introduction}
\label{sec-intro}

The critical behaviour of the self-avoiding walk depends on the spatial dimension
$d$.  The upper critical dimension is 4, and for $d \geq 5$ the lace
expansion has been used to prove that the asymptotic behaviour is
Gaussian \cite{BS85,Hara08,HHS03,HS92a,Slad06}.  In particular,
for the \emph{strictly}
self-avoiding walk in dimensions $d \ge 5$, the critical
two-point function has $|x|^{-(d-2+\eta)}$ decay with critical
exponent $\eta=0$, both for spread-out walks
\cite{CS14,HHS03} and for the nearest-neighbour walk \cite{Hara08}.  For
$d=3$, the problem remains completely unsolved from a mathematical
point of view, but numerical and other evidence
provides convincing evidence that the behaviour is not Gaussian.
In particular, numerical values of the critical exponents $\gamma$ and $\nu$
\cite{Clis10,SBB11}, together
with Fisher's relation $\gamma = (2-\eta) \nu$, indicate that the critical
two-point function has approximate decay $|x|^{-1.031}$ for $d=3$.
For $d=2$, the critical two-point function is
predicted to decay as $|x|^{-5/24}$ \cite{Nien82}, and recent work
suggests that the scaling behaviour is described by ${\rm
SLE}_{8/3}$ \cite{LSW04}, but neither has been proved.  The case of
$d=1$ is of interest for weakly self-avoiding walk, where a fairly
complete understanding has been obtained \cite{Holl09}.
More about mathematical results for self-avoiding walk can be found in
\cite{BDGS12,MS93}.

In the present paper, we prove that
the critical
two-point function of the continuous-time weakly self-avoiding walk
is asymptotic to a multiple of $|x|^{-2}$ as $|x|\to\infty$, in
dimension $d = 4$.
This is a statement that the critical exponent $\eta$
exists and is equal to zero.
The proof is based on a rigorous
renormalisation group method;
a summary of the method and proof is given in \cite{BS11}.
Early indications of the critical nature of the dimension $d=4$
were given in \cite{ACF83,BFF84}, following proofs of triviality of $\phi^4$
field theory above dimension 4 \cite{Aize82,Froh82}.

Logarithmic corrections to scaling are common in statistical mechanical
models at the upper critical dimension, and are predicted for the susceptibility
and correlation length and several other interesting quantities
\cite{LK69,BGZ73,WR73}, but not for the leading decay of the critical two-point function
of the 4-dimensional self-avoiding walk.
In \cite{BBS-saw4-log}, it is proved that the
susceptibility of the 4-dimensional weakly self-avoiding walk does have
a logarithmic correction to scaling, with exponent $\frac 14$.
We now extend the methods of \cite{BBS-saw4-log} to study the critical two-point function.

We use an integral representation to rewrite the two-point function
of the continuous-time weakly self-avoiding walk as the two-point
function of a supersymmetric field theory, and apply a rigorous
renormalisation group argument to analyse the field theory.

Our proof involves an extension of the ideas and structure developed
in \cite{BBS-saw4-log}, and to avoid repetition we refer below
frequently to \cite{BBS-saw4-log} for ideas and notation that apply
without modification to our present purpose.  A feature present here
but not in \cite{BBS-saw4-log} is the use of a complex observable
field $\sigma$; this requires aspects of
\cite{BS-rg-loc,BBS-rg-pt,BS-rg-IE,BS-rg-step} concerning observables
that were not used in \cite{BBS-saw4-log}. A similar extension was
used to study correlations of the dipole gas in \cite{DH92}.

Our general approach applies more widely. In \cite{ST-phi4}, it has
been extended to prove existence of logarithmic corrections to scaling
for 4-dimensional critical networks of weakly self-avoiding walks, and
for critical correlation functions of the 4-dimensional $n$-component
$|\varphi|^4$ spin model.

\subsection{Main result}
\label{sec:ctwsaw}

We now define the two-point function for continuous time
weakly self-avoiding walk, and state our main result.
Fix a dimension $d >0$.  Let $X$ be the stochastic process on $\Zd$
with right-continuous sample paths, that takes its steps at the times
of the events of a rate-$2d$ Poisson process.  Steps are independent both
of the Poisson process and of all other steps, and are taken uniformly
at random to one of the $2d$ nearest neighbours of the current
position.  Let $E_a$ denote the corresponding expectation for the
process started at $X(0)=a$.  The \emph{local time} at $x$ up to time $T$ is
defined by $L_{x,T} = \int_0^T \1_{X(s)=x} ds$, and the \emph{self-intersection local
time} up to time $T$ is $I(T) = \sum_{x\in \Zd} L_{x,T}^{2}$.
The continuous-time weakly
self-avoiding walk \emph{two-point function} is then defined by
\begin{equation}
\label{e:Gwsaw}
    G_{g,\nu}(a,b)
    =
    \int_0^\infty
    E_{a} \left(
    e^{-gI(T)}
    \1_{X(T)=b} \right)
    e^{- \nu T}
    dT,
\end{equation}
where $g>0$, and $\nu$ is a parameter (possibly negative)
chosen so that the integral converges.  By translation
invariance, $G_{g,\nu}(a,b)$ only depends on $a,b$ via $\pp -\qq$.
For $d=4$, the continuous-time weakly self-avoiding walk is identical
to the lattice Edwards model; see \cite[Section~10.1]{MS93}.

In \eqref{e:Gwsaw}, self-intersections are suppressed by the factor
$e^{-gI(T)}$.  In the limit $g \to \infty$, if $\nu$ is
simultaneously sent to $-\infty$ in a suitable $g$-dependent manner,
it is known that the limit of the two-point function \refeq{Gwsaw} is
a multiple of the two-point function of the standard discrete-time
strictly self-avoiding walk \cite{BDS12}.  The model defined by
\refeq{Gwsaw} is predicted to be in the same universality class as the
strictly self-avoiding walk for all $g>0$.  Our analysis is restricted
to small $g>0$.

The \emph{susceptibility} is defined by
\begin{equation}
\label{e:suscept-def}
    \chi_{g}(\nu) =   \sum_{b\in \Zd}  G_{g,\nu}(a,b)
    ,
\end{equation}
and the critical value $\nu_c(g)$ is defined by $\nu_{c} (g)=\inf\{\nu
\in \R : \chi_g(\nu)<\infty\}$.  It is proved in
\cite[Lemma~\ref{log-lem:csub}]{BBS-saw4-log} that $\nu_c=\nu_c(g,d)
\in (-\infty,0]$ for all $g>0$ and $d>0$, and that moreover
\begin{equation}
\label{e:chi-nuc}
    \text{$\chi_g(\nu) < \infty$ \; if and only if \; $\nu > \nu_c$}.
\end{equation}
Moreover, it is shown in \cite[Theorem~\ref{log-thm:nuc}]{BBS-saw4-log} that for
$d = 4$, as $g \downarrow 0$, \begin{equation} \label{e:nucasy}
\nu_c(g) = - {\sf a} g (1+O(g)), \end{equation} where the positive constant
${\sf a}$ is given by ${\sf a} = -2\Delta^{-1}_{00}$.  In particular,
\refeq{nucasy} implies that $\nu_c(g) <0$ for small positive $g$.

Our main result is the following theorem which gives the decay of the
critical two-point function in dimension $4$, for sufficiently small
$g$.

\begin{theorem}
\label{thm:wsaw4} Let $d = 4$.  There exists $\delta >0$ such that for
each $g \in (0,\delta)$ there exists $c(g)= (2\pi)^{-2}(1+O(g))$ such
that as $|\pp -\qq|\to \infty$,
\begin{equation}
\label{e:Gasy}
    G_{g,\nu_c(g)}(a,b) = \frac{c(g)}{|\pp -\qq|^{2}}
    \left( 1 + O\left( \frac{1}{\log |\pp-\qq|} \right) \right).
\end{equation}
\end{theorem}

In \cite{BS11}, an extension of Theorem~\ref{thm:wsaw4} states that
the critical two-point function has decay $|a-b|^{2-d}$ for all
dimensions $d\ge 4$, but \cite{BS11} provides only a sketch of proof.
Our principal interest is the critical dimension $d=4$, and we provide
the details of the proof for $d=4$ here.  The restriction to $d=4$
avoids additional complications required to handle general high
dimensions.  We intend to return to the general case in a future
publication.

We define the Laplacian $\Delta$ on $\Zd$ by $(\Delta f)_x = \sum_{e :
|e|=1} (f_{x+e}-f_x)$.  For $g=0$ and $\nu \ge 0$, the two-point
function is given by $G_{0,\nu}(a,b)=(-\Delta + \nu)^{-1}_{ab}$, and
$\nu_c(0)=0$.  Theorem~\ref{thm:wsaw4} proves that for $d = 4$ and
small positive $g$, the critical two-point function has the same
$|a-b|^{-2}$ decay as the lattice Green function $-\Delta^{-1}_{ab}$
on $\Z^4$.  In contrast, for $\nu > \nu_c(g)$, $G_{g,\nu}(a,b)$ decays
exponentially as $|\pp-\qq|\to\infty$; an elementary proof is sketched
below in Proposition~\ref{prop:Glims}.

In \cite[Section~\ref{log-sec:chvar}]{BBS-saw4-log}, it is shown that
the susceptibility of the weakly self-avoiding walk is \emph{equal} to
the susceptibility of the simple random walk with renormalised
diffusion constant (field strength) $1+z_0$ and killing rate (mass)
$m^2$, with $z_0$ and $m^2$ functions of $g,\varepsilon$ with
$\varepsilon =\nu-\nu_c(g)$.  More precisely,
\begin{equation}
    \chi_g(\nu)
    = \frac{1+ z_0}
    {m^2},
\end{equation}
with $z_0=O(g)$ and with $m^2$ asymptotic to a multiple of
$\varepsilon (\log \varepsilon^{-1})^{-1/4}$ as $\varepsilon\downarrow
0$.  In the proof of Theorem~\ref{thm:wsaw4} we extend this
correspondence and prove that (with $z_0=z_0(g,0)$)
\begin{equation}
    G_{g,\nu_c(g)}(a,b) \sim (1+z_0)
    (-\Delta_{\Z^d})^{-1}_{ab}
    \quad (|a-b|\to\infty),
\end{equation}
which shows that the critical interacting two-point function is
asymptotic to the critical non-interacting two-point function, with
the \emph{same} renormalised diffusion constant.

The proof of Theorem~\ref{thm:wsaw4} uses a supersymmetric integral
representation for the two-point function, which requires us to work
first in finite volume and with $\nu > \nu_c$.  Because of this, our
analysis initially stays slightly away from the critical point.  A
related issue is that the Laplacian annihilates constants in finite
volume, and hence is not invertible without the addition of some mass
term.  Ultimately, we first take the infinite volume limit with
$\nu>\nu_c$, and then let $\nu \downarrow \nu_c$.

A variant of the 4-dimensional Edwards model was analysed in
\cite{IM94} using a renormalisation group method.  Although this
variant is not a model of walks taking steps in a lattice, it is
presumably in the same universality class as the 4-dimensional
self-avoiding walk, and the results of \cite{IM94} are of a similar
nature to ours.  For the 4-dimensional $\varphi^4$ model, related
results were obtained via block spin renormalisation in
\cite{GK85,GK86,Hara87,HT87}, and via partially renormalised phase
space expansion in \cite{FMRS87}.  Our method is applied to the
$n$-component $|\varphi|^4$ model in \cite{BBS-phi4-log,ST-phi4}.

\section{Integral representation for the two-point function}

The proof of Theorem~\ref{thm:wsaw4} is based on an integral representation
for a finite volume approximation of the two-point function.
To discuss this, we first show how the two-point function can be approximated
by a two-point function on a finite torus.

\subsection{Finite volume approximation}
\label{sec:fv}

Let $L \ge 3$ and $N \ge 1$ be integers, and let $\Lambda =\Lambda_N = \Z^d/ L^N\Z^d$ denote
the discrete torus of side $L^N$.  We are ultimately interested in the limit
$N \to \infty$, and regard $\Lambda_N$ as a finite volume approximation to $\Zd$.
It is convenient at times to
consider $\Lambda_N$ to be a box (approximately) centred at the origin in $\Z^d$,
without identifying opposite sides to create the torus.
For fixed $a,b\in \Zd$, we can then regard $a,b$ as points in $\Lambda_N$
provided $N$ is large enough, and we make this identification
throughout the paper. In particular, we tacitly assume that $N$ is sufficiently large
to contain given $a,b$.

For $a,b\in \Lambda_N$, let
\begin{equation}
\label{e:GwsawLambda}
    G_{N,g,\nu}(a,b)
    =
    \int_0^\infty
    E_{a}^{\Lambda_N} \left(
    e^{-gI(T)}
    \1_{X(T)=b} \right)
    e^{- \nu T}
    dT,
\end{equation}
where $E_{a}^{\Lambda_N}$ denotes the continuous-time simple random walk
on the torus $\Lambda_N$, started from the point $a$.
By the Cauchy--Schwarz inequality,
$T= \sum_{x\in \Lambda} L_T^x \le (|\Lambda|I(T))^{1/2}$
and hence $I(T) \ge T^2/|\Lambda|$, from which we conclude that
the integral \refeq{GwsawLambda} is finite for all $g>0$ and $\nu\in\R$.
The following proposition shows that the infinite volume two-point
function \eqref{e:Gwsaw} can be approximated by the finite volume
two-point function \eqref{e:GwsawLambda}, and that it is
possible to study the critical two-point function on $\Zd$ in the double limit
in which first $N \to \infty$ and then $\nu \downarrow \nu_c$.

\begin{prop}
\label{prop:Glims}
Let $ d >0$, $g>0$, and $\nu > \nu_c(g)$.  Then
$G_\nu(a,b)$ decays exponentially in $|\pp -\qq|$, and
\begin{equation}
\label{e:Givl}
    G_{g,\nu}(a,b)
    =
    \lim_{N \to \infty} G_{N,g,\nu}(a,b).
\end{equation}
Also, for all $\nu \ge \nu_c(g)$,
\begin{equation}
\label{e:Givlc}
    G_{g,\nu}(a,b)
    =
    \lim_{\nu' \downarrow \nu} \lim_{N \to \infty}
    G_{N,g,\nu'}(a,b).
\end{equation}
\end{prop}

\begin{proof}
We fix $g>0$ and drop it from the notation.  Once we prove
\refeq{Givl}, \refeq{Givlc} then follows because, by monotone
convergence, $G_{g,\nu}(a,b)$ is right continuous for $\nu \ge
\nu_{c}(g)$.

Let $c_T(a,b) = E_{a} (e^{-gI(T)} \1_{X(T)=b} )$ and $c_{N,T}(a,b) =
E_{a}^{\Lambda_N} (e^{-gI(T)} \1_{X(T)=b} )$.  Fix $\nu > \nu' >
\nu_c$ and $S>0$.  By the triangle inequality,
\begin{equation}
\begin{aligned}
\label{e:G-three-contributions}
|\, G_{\nu}(a,b) - G_{N,\nu}(a,b) \, |
&\leq
\int_{0}^{S}  |c_{T}(a,b)-c_{N,T}(a,b)| e^{-\nu T} dT
+  \int_{S}^\infty c_{T}(a,b) e^{-\nu T} dT
\\
&\qquad
+  \int_{S}^{\infty}  c_{N,T}(a,b) e^{-\nu T}dT
.
\end{aligned}
\end{equation}

For the analysis of the right-hand side of \eqref{e:G-three-contributions}, we
define $\chi_{N}(\nu) = \sum_{b \in\Lambda}G_{N,\nu}(a,b)$, and recall from
\cite[Lemma~\ref{log-lem:suscept-finvol}]{BBS-saw4-log} that
$\chi_{N}(\nu) \leq \chi(\nu)$.  From this it follows that
\begin{align}\label{e:W_limsup_assumption}
  \limsup_{N \to \infty} G_{N, \nu'}(a,b)
  \le
  \limsup_{N \to \infty} \chi_{N}(\nu')
  \le \chi(\nu')
  < \infty.
\end{align}
Let $\delta = \nu - \nu' > 0$.  Then
\begin{align}
  \int_{S}^\infty c_{T}(a,b) e^{-\nu T} dT &\leq e^{-\delta S} G_{\nu'}(a,b),
  \\
  \limsup_{N \to \infty} \int_{S}^\infty c_{N,T}(a,b) e^{-\nu T} dT
  &\leq e^{-\delta S} \limsup_{N \to \infty} G_{N,\nu'}(a,b).
\end{align}
This shows that the last two terms in \eqref{e:G-three-contributions}
can be made as small as desired, uniformly in $N$, by choosing $S$
large.

To estimate the first contribution, let $(Y_t)_{t\geq 0}$ be a
rate-$2d$ Poisson process with corresponding probability measure $P$.
Since contributions to the difference $|c_{T}-c_{N,T}|$  only
arise from walks that reach the inner boundary $\partial \Lambda$ of the torus
(identified with a subset of $\Zd$ so that it does have a boundary),
for any $0 \le T \le S$ we have
\begin{align}
  |c_{T}(a,b)-c_{N,T}(a,b)|
  &\leq
  E_a\left(e^{-gI(T)} \1_{\left\{X([0,T])\cap \partial\Lambda\neq\varnothing\right\}}
  \right)
  + E_{a}^{\Lambda}\left(e^{-gI(T)} \1_{\left\{X([0,T])\cap \partial\Lambda\neq\varnothing\right\}}\right)
 \nnb
  &\leq P_a \left\{X\big([0,T]\big)\cap \partial\Lambda\neq\varnothing\right\}
  + P_{a}^{\Lambda} \left\{X\big([0,T]\big)\cap \partial\Lambda\neq\varnothing\right\}
 \nnb
  &\leq 2 P\left\{Y_{T} \geq \diam\Lambda\right\}
  \leq 2 P\left\{Y_{S} \geq \diam\Lambda\right\},
\end{align}
and the right-hand side goes to zero as $N \to\infty$ with $S$ fixed.
By this estimate, the first integral in
\eqref{e:G-three-contributions} converges to $0$ as $N\to\infty$, for
any fixed $S$.  This completes the proof of \refeq{Givl} and hence of
\refeq{Givlc}.

We do not use the exponential decay in this paper, and its proof is standard,
so we only sketch the argument.
Given any $\alpha >0$, we write $x=b-a$ and make the division
\begin{equation}
\lbeq{Gsplit}
    G_\nu(a,b) = \int_0^{\alpha |x|} c_T(a,b) e^{-\nu T} dT
    +
    \int_{\alpha |x|}^\infty c_T(a,b) e^{-\nu T} dT.
\end{equation}
For the second integral on the right-hand side, we
set $c_T=\sum_{b \in \Z^d} c_T(a,b)$ and use $c_T(a,b) \le c_T$.
It can be shown that $c_T$ obeys $c_{S+T}\le c_S c_T$ and that this implies
$c_T^{1/T} \to e^{\nu_c}$
as $T \to \infty$, from which we see that $c_T(a,b)
\le C_\epsilon e^{(\nu_c+\epsilon)T}$ for any $\epsilon >0$.  This gives exponential
decay in $x$ for the second integral.  For the first integral,
we recall the Chernoff estimate for the Poisson distribution, in the form that
if $X$ is Poisson with mean $\lambda$ and
$k >\lambda$, then $P(X>k) \le e^{-\lambda}(e\lambda/k)^k$.
Since a walk can travel from $a$ to $b$ in time $T$ only if the number of steps
taken is at least $x$, it follows from the Chernoff bound that if
$2d\alpha <1$ and  $T \le \alpha |x|$ then
\begin{equation}
    c_T(a,b) \le P(Y_T \ge |x|) \le e^{-2dT} (2d Te)^{|x|} |x|^{-|x|}
    \le
    (2de \alpha)^{|x|}.
\end{equation}
By choosing $\alpha$ sufficiently small depending on $\nu$ (recall that $\nu<0$ is possible),
we see that the first term on the right-hand side
of \refeq{Gsplit} also exhibits exponential decay in $x$.
\end{proof}

\subsection{Integral representation}

We use the supersymmetric integral representation for the two-point function
discussed in detail in \cite[Section~\ref{log-sec:intrep}]{BBS-saw4-log}.
We refer to that discussion for further details, notation, and definitions, and
here we recall the minimum needed for our present purposes.

In terms of the complex
boson field $\phi,\bar\phi$ and conjugate fermion fields $\psi,\bar\psi$ introduced in
\cite[Section~\ref{log-sec:intrep}]{BBS-saw4-log}, and using the same notation,
for $x \in \Lambda$
we define the differential forms
\begin{equation}
    \tau_x
    = \phi_x \bar\phi_x
    + \psi_x  \wedge \bar\psi_x
    ,
\end{equation}
\begin{equation}
\label{e:addDelta}
    \tau_{\Delta,x}
    =
    \frac 12 \Big(
    \phi_{x} (- \Delta \bar{\phi})_{x} + (- \Delta \phi)_{x} \bar{\phi}_{x} +
    \psi_{x} \wedge (- \Delta \bar{\psi})_{x} + (- \Delta \psi)_{x} \wedge \bar{\psi}_{x}
    \Big),
\end{equation}
where $\Delta=\Delta_\Lambda$ is the lattice Laplacian on $\Lambda$ given by $\Delta
\phi_{x} = \sum_{y: |y-x|=1} (\phi_{y} - \phi_{x})$.  Here $\wedge$ denotes
the wedge product; we drop the wedge from the notation subsequently with the
understanding that forms are always multiplied using this anti-commutative product.

Let $\Ex_C$  denote the Gaussian super-expectation with covariance
matrix $C$, as defined in \cite[Definition~\ref{log-def:ExC}]{BBS-saw4-log}.
In \cite[\eqref{log-e:GG2}--\eqref{log-e:Gmgnzdef}]{BBS-saw4-log},
it is shown that for $N<\infty$, $g>0$, $\nu \in \R$, $m^2>0$, and
$z_0>-1$,
\begin{equation}
\lbeq{GNEx}
    G_{N,g,\nu}(a,b)
    =
    (1+z_0) \Ex_C \left( e^{-U_0(\Lambda)} \bar\phi_a\phi_b \right)
    ,
\end{equation}
where $C=(-\Delta + m^2)^{-1}$,
\begin{align}
    \label{e:Vtil0def}
    \Vbulk_{0} (\Lambda)
    &
    =
    \sum_{x\in\Lambda}
    \big(g_{0} \tau_x^2 + \nu_{0} \tau_x + z_{0}\tau_{\Delta ,x}\big)
    ,
\end{align}
and
\begin{equation}
  \label{e:gg0}
  g_0 = g(1+z_0)^2, \quad \quad
  \nu_0 = (1+z_0)\nu-m^2
  .
\end{equation}
The identity \refeq{GNEx} is a rewriting of an identity from
\cite{BEI92,BI03d} that was
inspired by \cite{McKa80,PS80}; see also \cite[Theorem~5.1]{BIS09} for a
self-contained proof.

In \cite{BBS-saw4-log}, we also use \eqref{e:GNEx}, but there write $V$ instead of $U$.
In the present paper, we use $V$ for an extension of $U$ that incorporates also an \emph{observable field},
discussed next.

\subsection{Observable field}
\label{sec:of}

We introduce an external field $\sigma \in \C$ and define
\begin{equation}
    \label{e:V0def}
    V_{0} (\Lambda )
    =
    \Vbulk_{0} (\Lambda)
    -
    \sigma \phib_{\pp} - \sigmab \phi_{\qq}.
\end{equation}
We refer to $\sigma$ as the \emph{observable field}.
Then we can compute the two-point function using the identity
\begin{align}
    G_{N,g,\nu}(a,b)
    &=
    (1+z_0)
    \frac{\partial^2 }{\partial \sigma  \partial \sigmab }
    \Big|_{0}
    \Ex_C  e^{-V_0(\Lambda)}
    ,
\label{e:generating-fn}
\end{align}
which follows from \eqref{e:GNEx}.
To prove Theorem~\ref{thm:wsaw4}, we analyse the derivative of the
Gaussian super-expectation on the right-hand side of \refeq{generating-fn},
without making further
reference to its connection with self-avoiding walks.

An external field is also employed to analyse the susceptibility in
\cite[Section~\ref{log-sec:ga}]{BBS-saw4-log}, but in a different way.
There the external field is a test function $J:\Lambda \to \R$,
and $\Vbulk_0(\Lambda)$ becomes replaced by $\Vbulk_0(\Lambda) -
\sum_{x \in \Lambda}(J_x\bar\phi_x + \bar J_x \phi_x)$.  In
\cite{BBS-saw4-log} the interest is in the \emph{constant} external
field $J_x =1$ for all $x \in \Lambda$, and the macroscopic regularity
of this test function is important.  Here, in contrast, \refeq{V0def}
corresponds to setting $J_x = \sigma \1_{x=a}$ and $\bar J_x =
\bar\sigma \1_{x=b}$ (so the two are not precisely complex
conjugates).  To work with such a \emph{singular} external field, we
use a different analysis based on ideas prepared in
\cite{BS-rg-loc,BS-rg-IE,BS-rg-step}.  It would be desirable to allow
all coupling constants to be spatially varying, not just the external
field. This extension has been achieved for hierarchical models in
\cite{ACG13}.

Our attention to the dependence on the external field is quite
limited: we only wish to compute the derivative \refeq{generating-fn},
and as such we make no use of any functional dependence on
$\sigma,\bar\sigma$ beyond expansion to second order, i.e., including
terms of order $1,\sigma,\bar\sigma, \sigma\bar\sigma$.  We formalise
this notion by identifying quantities with the same expansion to
second order, as follows.  Recall the space $\Ncal$ of even
differential forms introduced in
\cite[Section~\ref{log-sec:intforms}]{BBS-saw4-log}, which we now
denote instead by $\Ncal^\varnothing$.  As in
\cite[\eqref{log-e:psipsib}]{BBS-saw4-log}, an element of
$\Ncal^\varnothing$ has the form
\begin{equation} \label{e:psipsib}
  \sum_{x,y} F_{x,y}(\phi,\bar\phi) \psi^x \bar\psi^{y}
  .
\end{equation}
We extend this notion by now allowing the coefficients $F_{x,y}$ to be
functions of the external field $\sigma,\bar\sigma$ as well as of the
boson field $\phi,\bar\phi$.  Let $\Ncal$ be the resulting algebra of
differential forms.  Let $\Ical$ denote the ideal in $\Ncal$
consisting of those elements of $\Ncal$ whose expansion to second
order in the external field is zero.  The quotient algebra
$\Ncal/\Ical$ then has the direct sum decomposition
\begin{equation}
\label{e:Ncaldecomp}
    \Ncal/\Ical = \Ncal^\varnothing \oplus \Ncal^a \oplus \Ncal^b \oplus \Ncal^{ab},
\end{equation}
where elements of $\Ncal^a, \Ncal^b , \Ncal^{ab}$ are respectively
given by elements of $\Ncal^\varnothing$ multiplied by $\sigma$, by
$\bar\sigma$, and by $\sigma\bar\sigma$.  For example, $\phi_x
\bar\phi_y \psi_x \bar\psi_x \in \Ncal^\varnothing$, and $\sigma
\bar\phi_x \in \Ncal^a$.  There are canonical projections $\pi_\alpha:
\Ncal \to \Ncal^\alpha$ for $\alpha \in \{\varnothing, a, b, ab\}$.
We use the abbreviation $\pi_*=1-\pi_\varnothing =
\pi_a+\pi_b+\pi_{ab}$.  The quotient space is used also in
\cite{BS-rg-loc,BS-rg-IE,BS-rg-step}, e.g., around
\cite[\eqref{loc-e:1Ncaldecomp}]{BS-rg-loc}.  Since we have no further
use of $\Ncal$, to simplify the notation we henceforth write simply
$\Ncal$ instead of $\Ncal/\Ical$.  As functions of the external field,
elements of $\Ncal$ are then polynomials of degree at most $2$, by
definition.  For example, we identify $e^{\sigma\bar\phi_a +
\bar\sigma\phi_b}$ and $1+\sigma\bar\phi_a + \bar\sigma\phi_b +
\sigma\bar\sigma\bar\phi_a\phi_b$, as both are elements of the same
equivalence class in the quotient space.

\section{Renormalisation group map}
\label{sec:rgmap}

In this section, we sketch only the most important ingredients of our
renormalisation group method from
\cite{BBS-rg-pt,BS-rg-IE,BS-rg-step,BBS-saw4-log}.  A more detailed
introduction is given in \cite{BBS-saw4-log} (see also
\cite{BBS-phi4-log,BS11}).

\subsection{Progressive Gaussian integration}
\label{sec:pgi}

We use decompositions of the covariances
$C=(-\Delta_{\Lambda_N}+m^2)^{-1}$ and $(-\Delta_{\Z^4}+m^2)^{-1}$ for
the torus and $\Z^4$, respectively, as discussed in
\cite[Section~\ref{log-sec:decomposition}]{BBS-saw4-log}, and we use
the same notation as in \cite{BBS-saw4-log}.  These decompositions
take the form
\begin{align}
\label{e:Zddecomp}
    (-\Delta_{\Z^4}+m^2)^{-1}
    &=
    \sum_{j=1}^\infty C_j
    \quad\quad
    (m^2 \in [0,\delta)),
\\
\lbeq{Cdef}
    C=
    (-\Delta_{\Lambda_N}+m^2)^{-1}
    &= \sum_{j=1}^{N-1}C_j+C_{N,N}
    \quad\quad
    (m^2 \in (0,\delta)),
\end{align}
where the covariance $C_{N,N}$ is special because of the effect of the
torus.
The particular
finite-range decomposition we use is developed in
\cite{Baue13a,BGM04}, with properties given in \cite{BBS-rg-pt}.
  The finite-range condition is the statement that
$C_{j;x,y}=0$ when $|x-y|\ge \frac 12 L^j$; this condition
is important for results we use from \cite{BS-rg-IE,BS-rg-step}.
As discussed in
\cite[Section~\ref{log-sec:decomposition}]{BBS-saw4-log}, the Gaussian
super-expectation of $F \in \Ncal$ can be carried out progressively,
via the identity
\begin{equation}
    \label{e:progressive}
    \Ex_{C}F
    =
    \big( \Ex_{C_{N,N}}\circ \Ex_{C_{N-1}}\theta \circ \cdots
    \circ \Ex_{C_{1}}\theta\big) F
    .
\end{equation}

The external field $\sigma,\bar\sigma$ is treated as a constant by the
super-expectation.  To compute $\Ex_C e^{-V_0(\Lambda)}$ of
\refeq{generating-fn}, we use \eqref{e:progressive}, and define
\begin{equation}
\label{e:Z0def}
  Z_0 = e^{-V_0(\Lambda)}, \quad Z_{j+1} = \Ex_{C_{j+1}}\theta Z_j \;\;\;
  (j<N).
\end{equation}
For $j+1=N$, we interpret the convolution $\Ex_{C_{j+1}}\theta$ as the
expectation\ $\Ex_{C_{N,N}}$, i.e., the last covariance is taken to be
the one appropriate for the torus $\Lambda_N$.
Then the desired expectation is given by $Z_N^0(0)$, where the
superscript $0$ denotes projection onto the degree-$0$ part of the
differential form (i.e., the fermion field is set to $0$) and the
argument $0$ means that the boson field is evaluated at $\phi=0$.
Thus we are led to study the recursion $Z_j \mapsto Z_{j+1}$.  By
\refeq{generating-fn}, the two-point function is given by
\begin{align}
    G_{N,g,\nu}(a,b)
    &=
    (1+z_0)
    Z_{N;\sigma\bar\sigma}^0(0)
    ,
\label{e:GZN}
\end{align}
where $F_{\sigma\bar\sigma}\in\Ncal^\varnothing$ denotes the coefficient of
$\sigma\bar\sigma$ in $F\in \Ncal$, i.e., $\pi_{ab}Z_N^0 = \sigma\sigmab Z^0_{N;\sigma\bar\sigma}$.

\subsection{The interaction functional}
\label{sec:formofint}

Let $\Qcal^{(0)}$ and $\Qcal^{(1)}$ respectively denote the vector
space of local polynomials of the form
\begin{align}
    \label{e:Vterms}
    V^{(0)}
    &=
    g \tau^{2} + \nu \tau + z \tau_{\Delta}
    -
    \lambda_{\pp}\1_{\pp} \,\sigmaa \bar{\phi}
    -
    \lambda_{\qq}\1_{\qq} \,\sigmab \phi
    ,
    \\
    V^{(1)}
    &=
    V^{(0)}
    -
     \textstyle{\frac 12} \sigmaa\sigmab(q_a\1_{a} + q_b\1_{b})
    ,
\end{align}
where $g,\nu,z \in \R$, $\lambda_a,\lambda_b,q_a,q_b \in \C$, and the indicator
functions are defined by the Kronecker delta $\1_{a,x}=\delta_{a,x}$.
(We believe that in fact only real coupling constants
$\lambda_a,\lambda_b,q_a,q_b$ are required, but we did not prove this and
it costs us nothing to permit complex coupling constants.)
The terms involving $\sigma$ are referred to as \emph{observables},
while the terms involving $\tau^2$, $\tau$, and $\tau_\Delta$ are
\emph{bulk terms}.
We frequently identify elements of $\Qcal^{(0)}$
and $\Qcal^{(1)}$ as sequences $V^{(0)} =
(g,\nu,z,\lambda_a,\lambda_b)$, $V^{(1)} =
(g,\nu,z,\lambda_a,\lambda_b,q_a,q_b)$, and typically write $U = \pi_\varnothing V =
(g,\nu,z)$.

Recall from \cite[Section~\ref{log-sec:polymers}]{BBS-saw4-log}
the set ${\cal B}_j$  of scale-$j$ blocks, and the set ${\cal P}_j$
of scale-$j$ polymers in $\Lambda$.  We also recall from
\cite[Section~\ref{log-sec:if}]{BBS-saw4-log} the interaction functional
$I_j : \Qcal^{(0)} \times \Pcal_j \to \Ncal$ defined
for $B \in {\cal B}_j$, $X \in {\cal P}_j$,
and $V\in \Qcal^{(0)}$ by
\begin{align}
\label{e:Idef}
    I_{j} (V,B)
    & = e^{-V(B)}(1+W_j(V,B)),
    \quad\quad
    I_{j} (V,X)
    =
    \prod_{B\in \Bcal_{j}}I_{j} (V,B)
    ,
\end{align}
where $W_j$ is an explicit quadratic function of $V$ defined in \cite{BBS-rg-pt}.  In
particular, $W_0=0$. We often write simply $I_{j} (X) = I_{j} (V,X)$.
By \eqref{e:Idef},
$I_{0} (V,X)=e^{-V (X)}$ for all $X \subset \Lambda$, with
$V(X)=\sum_{x\in X}V_x$.

Motivation for the definition
\eqref{e:Idef} is given in
\cite[Section~\ref{pt-sec:WPjobs}]{BBS-rg-pt}.  In the present paper,
we do not give the details of the definitions of  $W_j$
and $I_{j}$
since we do not need them here. They are, however,
important in \cite{BBS-rg-pt,BS-rg-IE,BS-rg-step} and we
rely on results from those references.
The
$V$ domain of $I_j$ is larger here than in \cite{BBS-saw4-log}, due
to the presence of observables, but the larger domain is permitted and present
in the analysis of \cite{BBS-rg-pt,BS-rg-IE,BS-rg-step}.

\subsection{Renormalisation group coordinates}
\label{sec:rgcoord}

Given $F_1, F_2 :{\cal P}_j \to \Ncal$, we define the \emph{circle product}
$F_1 \circ F_2 :{\cal P}_j \to \Ncal$ by
\begin{equation}
    (F_1 \circ F_2)(Y) = \sum_{X\in {\cal P}_j: X\subset Y} F_1(X) F_2(Y \setminus X)
    \quad\quad (Y \in \Pcal_j).
\end{equation}
The terms $X=\varnothing$ and $X=\Lambda$ are included
in the summation on the right-hand side, and we demand that all functions
$F : \Pcal_j \to \Ncal$ obey $F(\varnothing)=1$.
The circle product depends on the scale~$j$, is associative, and is also commutative
due to our restriction in $\Ncal$ to forms of even degree.
Its identity element is $\1_\varnothing$, defined by
$\1_\varnothing(X) = 1$ if $X$ is empty, and otherwise $\1_\varnothing(X) = 0$.

In the definition of $I_{0}$ we set $V=V_{0}$, with $V_0$
defined in \eqref{e:V0def}, so that $I_{0} (X) = I_{0}
(V_0,X)=e^{-V_{0} (X)}$ for all $X \subset \Lambda$.  Let $K_0 :
\Pcal_0 \to \Ncal$ be defined by $K_0 = \1_{\varnothing}$, and set
$q_0=0$.  Then $Z_{0}=I_0(V_0,\Lambda)$ of \eqref{e:Z0def} is also
given by
\begin{equation}
\label{e:Zinit}
    Z_0 = I_0(\Lambda) =  e^{q_0\sigma\bar\sigma}(I_0 \circ K_0)(\Lambda)
    .
\end{equation}
Our strategy is to define $q_j \in \C$,
$V_j \in \Qcal^{(0)}$,
$K_{j}:\Pcal_{j}\rightarrow \Ncal$, and set $I_j=I_j(V_j)$,
so as to maintain this form as
\begin{equation}\label{e:Zjcirc}
    Z_j =
    e^{q_j\sigma\bar\sigma}
    (I_j \circ K_j)(\Lambda)
    \quad\quad
    (0 \le j \le N)
\end{equation}
in the recursion $Z_j \mapsto Z_{j+1} = \Ex_{C_{j+1}}\theta Z_j$ of
\eqref{e:Z0def}, with the initial condition given by \eqref{e:Zinit}.
At the final scale $j=N$,
the only two polymers are the single block $\Lambda=\Lambda_N$ and
the empty set $\varnothing$,
and since $I_j(\varnothing)=K_j(\varnothing)=1$, by assumption,
\eqref{e:Zjcirc} simply reads
\begin{equation}\label{e:ZNcirc}
  Z_N
  =
  e^{q_N\sigma\bar\sigma}
  (I_N \circ K_N)(\Lambda)
  =
  e^{q_N\sigma\bar\sigma} (I_N(\Lambda)+K_N(\Lambda))
  .
\end{equation}
If we set  $\delta q_{j+1} = q_{j+1}-q_j$,
then \refeq{Zjcirc} can equivalently be written as
\begin{equation} \label{e:IcircKdq}
  \Ex_{C_{j+1}}\theta(I_j \circ K_j)(\Lambda)
  =
  e^{\delta q_{j+1}\sigma\bar\sigma}(I_{j+1} \circ K_{j+1})(\Lambda)
  .
\end{equation}
In view of \refeq{IcircKdq}, and since $I_j$ is determined by $V_j$,
we are led to study the \emph{renormalisation group map
\begin{equation} \lbeq{RGmap}
    (V_j,K_j) \mapsto (\delta q_{j+1},V_{j+1},K_{j+1}).
\end{equation}}%
The coupling constants of $V_j \in \Qcal^{(0)}$ are written as
$g_j,\nu_j,z_j,\lambda_{a,j},\lambda_{b,j}$.
Ultimately we express the two-point function in terms of the sequence
$(q_j)$, so this sequence is fundamentally important in the proof of
Theorem~\ref{thm:wsaw4}.  Our construction creates $\delta q_j$ as the
average
\begin{equation}
\lbeq{delqdef}
    \delta q_j=\frac 12 (\delta q_{a,j}+ \delta q_{b,j})
\end{equation}
of two sequences $\delta q_{a,j}$ and $\delta_{b,j}$ (see
\cite[\eqref{step-e:Rplusdef}]{BS-rg-step}).

\subsection{Renormalisation group map}

To implement the above strategy, given suitable
$V_j \in \Qcal^{(0)}$ and $K_j : \Pcal_j \to \Ncal$, we
define $\delta q_{j+1} \in \C$,
$V_{j+1}\in \Qcal^{(0)}$ and
$K_{j+1} : \Pcal_{j+1} \to \Ncal$ in such a way that
\begin{equation}
    \label{e:Kspace-objective}
    Z_{j+1}
     = \Ex_{C_{j+1}}\theta Z_j
    =
    e^{q_j\sigma\bar\sigma}
    \Ex_{C_{j+1}}\theta (I_j \circ K_j)(\Lambda)
    =
    e^{q_{j+1}\sigma\bar\sigma}
    (I_{j+1} \circ K_{j+1})(\Lambda)
    \quad (j<N)
    .
\end{equation}
Thus \refeq{Zjcirc} does retain its form under progressive integration.
We use the explicit choice for the renormalisation group map
\eqref{e:RGmap} that is given in \cite{BS-rg-step}, from now on.  This
choice achieves \eqref{e:Kspace-objective} for fixed $j<N$, assuming
that $(V_j,K_j)$ is in an appropriate domain, and it provides good
estimates for $(\delta q_{j+1},V_{j+1},K_{j+1})$.

To simplify the notation, we set $V_+ = (\delta q_+,
V_+^{\smash{(0)}}) \in \Qcal^{(1)}$ and write \refeq{RGmap} as $(V,K)
\mapsto (V_+,K_+)$.  We typically drop subscripts $j$ and write $+$ in
place of $j+1$, also leave the dependence of the maps on the mass
parameter $m^2$ of the covariance $(-\Delta+m^2)^{-1}$ implicit.

\subsection{Bulk flow}
\label{sec:bulk}

By \cite[\eqref{step-e:piVKplus}]{BS-rg-step}, the renormalisation
group map has the property
\begin{equation} \label{e:bulk}
  \pi_\varnothing V_+(V,K) = V_+(\pi_\varnothing V,\pi_\varnothing K), \quad
  \pi_\varnothing K_+(V,K) = K_+(\pi_\varnothing V,\pi_\varnothing K).
\end{equation}
Thus, under \eqref{e:RGmap},
the \emph{bulk coordinates} $(\pi_\varnothing V_j, \pi_\varnothing K_j)$ satisfy
a closed evolution equation of their own. We denote its evolution map by
$(V_+^\varnothing,K_+^\varnothing)$ and write $U = \pi_\varnothing V$.
Then \eqref{e:RGmap} reduces to
\begin{equation} \label{e:RGmapbulk}
  (U_{j+1}, \pi_\varnothing K_{j+1}) = (V_+^\varnothing(U_j,\pi_\varnothing K_j), K_+^\varnothing(U_j,\pi_\varnothing K_j)).
\end{equation}

The construction of a critical global
renormalisation group flow of the bulk coordinates \eqref{e:RGmapbulk} is achieved in
\cite{BBS-saw4-log}.  Namely, there is a construction of $(U_j,\pi_\varnothing K_j)$ for
$0 \le j \le N$ such that
\eqref{e:RGmapbulk} holds for all $0 \le j \le N$.
This construction provides detailed information about the
sequence $U_j$, and good estimates on $\pi_\varnothing K_j$, sufficient for studying the infinite
volume limit at the critical point.
In Section~\ref{sec:obflow}, we use this bulk flow to study observables.

It is convenient to change perspective on which variables are independent.
The
weakly self-avoiding walk has parameters $g,\nu$.
In \eqref{e:Vtil0def}, additional parameters $m^2$, $g_0$, $\nu_0$, $z_0$ were introduced.
For the moment we consider these as independent variables and do not consider $g,\nu$
directly.
The relation between $m^2,g_0,\nu_0,z_0$
and the original parameters $g,\nu$ is addressed in Section~\ref{sec:changevariables}.

To state the result about the bulk flow,
let $\gbar_j$ be the $(m^2,g_0)$-dependent sequence determined by
$\gbar_{j+1}=\gbar_j - \beta_j\gbar_j^2$ with $\gbar_0=g_0$ and
with $\beta_j=\beta_j(m^2)$ defined in \cite[\eqref{log-e:betadef}]{BBS-saw4-log}.
We also recall the sequence $\chi_j$ defined in
\cite[\eqref{log-e:chidef}]{BBS-saw4-log}, but its precise definition is not
important for our present needs.
It obeys $0 \le \chi_j \le 1$,
eventually decays exponentially when $m^2>0$, and is
identically equal to $1$ when $m^2=0$.  Also, by
\cite[Proposition~\ref{log-prop:approximate-flow}]{BBS-saw4-log}
and
\cite[\eqref{log-e:chigbd-bis}]{BBS-saw4-log} respectively,
\begin{equation}
  \lbeq{chiglim}
    \chi_j \gbar_j \leq O\left(\frac{g_0}{1+g_0j}\right)
    \quad
    \text{uniformly in $(m^2,g_0) \in [0,\delta)^2$},
\end{equation}
\begin{equation}
  \lbeq{chisum}
    \sum_{k=j}^\infty \chi_k \gbar_k^2 = O(\chi_j \gbar_j).
\end{equation}
Without multiplication by $\chi_j$,
the sequence $\gbar_j$ converges to $0$ when $m^2=0$ but not when $m^2>0$.
(To apply \eqref{e:Givlc}, in which the limit $\nu \downarrow \nu_c$
follows the limit $N \to \infty$,
we do consider limits $j \to\infty$ with $m^2>0$, corresponding
to $\nu > \nu_c$, to prove Theorem~\ref{thm:wsaw4}.)

The following theorem is a reduced version of
\cite[Proposition~\ref{log-prop:KjNbd}]{BBS-saw4-log}.
Some of its notation is explained after the statement.

\begin{theorem} \label{thm:bulkflow}
  Let $d=4$ and let $\delta>0$ be sufficiently small.
  There exist $M>0$ and an infinite sequence
  of functions $U_j=(g_j^c,\nu_j^c,z_j^c)$ of $(m^2,g_0) \in [0,\delta)^2$,
  independent of $N \in \N$, such that:
  \smallskip
  \\
  (i) assuming $\sigma=0$, given $N \in \N$,
  for initial conditions $U_0 = (g_0,\nu_0^c,z_0^c)$ with $g_0 \in (0,\delta)$, $K_0=\1_{\varnothing}$,
  and with mass $m^2 \in [0,\delta)$,
  a flow  $(U_j,K_j) \in \domRG_j^\varnothing$
  exists such that \eqref{e:RGmapbulk} holds for all $j+1 < N$,
  and given $m^2 \in [\delta L^{-2(N-1)}, \delta)$, also for $j+1=N$.
  Then, in particular,
  \begin{align}
    \lbeq{VVbar1-bulkflow}
    \|K_j\|_{\Wcal_j}
    =     \|\pi_\varnothing K_j\|_{\Wcal_j}
    &\leq M\chicCov_j
    \gbar_j^3
    \quad ( j \le N)
  \end{align}
  and $g_j^c = O(\gbar_j)$. In addition, $z_j^c = O(\chi_j \gbar_j)$ and $\nu_j = O(\chi_j L^{-2j} \gbar_j)$.

  \smallskip\noindent
  (ii) $z_0^c, \nu_0^c$ are continuous in $(m^2,g_0) \in [0,\delta)^2$.
\end{theorem}

The definition of the $\Wcal_j$ norm in \refeq{VVbar1-bulkflow} is
discussed at length in \cite{BS-rg-step}, and we do not repeat the
details here, as we now only need the fact that
\refeq{VVbar1-bulkflow} with $j=N$ implies that
\begin{equation}
\lbeq{Kgnull}
    |\pi_\varnothing K_N^0(\Lambda)|
    \le M
    \chi_N \gbar_N^3,
\end{equation}
uniformly in $m^2 \in [\delta L^{-2(N-1)}, \delta)$, as a consequence
of \cite[\eqref{step-e:Kg1}]{BS-rg-step}.

The $\Wcal_j=\Wcal_j(\sgen)$ norm depends on a parameter $\sgen =
(\mgen^2,\ggen) \in [0,\delta) \times (0,\delta)$.  Its significance
is discussed in \cite[Section~\ref{log-sec:flow-norms}]{BBS-saw4-log}.
In particular, useful choices of this parameter depend on the scale
$j$, as well as on approximate values of the mass parameter $m^2$ of
the covariance and the coupling constant $g_j$.
Throughout the paper, we use the convention that when the parameter
$\sgen$ is omitted, it is given by $\sgen = s_j =
(m^2,\ggen_j(m^2,g_0))$.  Here $\mgen^2=m^2$ is the mass parameter of
the covariance, and $\ggen = \ggen_j$ is defined in terms of the
initial condition $g_0$ by
\begin{equation} \label{e:ggendef}
  \ggen_j = \ggen_j(m^2,g_0)
  = \gbar_j(0,g_0) \1_{j \le j_m} + \gbar_{j_m}(0,g_0) \1_{j > j_m},
\end{equation}
where the \emph{mass scale} $j_m$ is the smallest integer $j$ such
that $L^{2j}m^2 \ge 1$.  By
\cite[Lemma~\ref{log-lem:gbarmcomp}]{BBS-saw4-log},
\begin{equation}
\lbeq{gbarggen}
    \ggen_j=\gbar_j+O(\gbar_j^2),
\end{equation}
so the two sequences are the same to
leading order.  However, $\ggen_j$ is more convenient for aspects of the analysis in \cite{BBS-saw4-log}.

The domain $\domRG_j^\varnothing = \domRG_j^\varnothing(\sgen)$ also
depends on $\sgen$ (with the same convention when the parameter is
omitted) and is defined as follows.  For the universal constant $C_\DV
\ge 2$ determined in \cite{BBS-saw4-log}, for $j<N$,
\begin{equation} \label{e:domRGbulk}
  \domRG_j^\varnothing(\sgen)
  = \{(g,\nu,z)\in \R^3 :
    C_{\DV}^{-1} \ggen <  g < C_{\DV} \ggen,
    \; L^{2j}|\nu|,|z| \le C_\DV \ggen \}
    \times B_{\Wcal_j^\varnothing}(\DVa\chigen_j\ggen^3).
\end{equation}
The first factor is the important stability
domain defined in \cite[\eqref{step-e:DV1}]{BS-rg-IE}, restricted to
the bulk coordinates and real scalars.
In the second factor,
$B_X(a)$ denotes the open ball of radius $a$
centred at the origin of the Banach space $X$, and  $\DVa$ is fixed in
\cite{BBS-saw4-log}; it can be taken to be $10M$ where $M$ is the constant of
Theorem~\ref{thm:bulkflow}.
Compared to \cite{BS-rg-step}, we have replaced $\chi^{3/2}$ by
$\chi$ for notational convenience.
The space $\Wcal^\varnothing$ is the restriction of $\Wcal$ to
elements $K \in \Wcal$ with $\pi_*K(X) = 0$ for all polymers
$X$. Since the renormalisation group acts triangularly, by
\eqref{e:bulk}, the distinction between $\Wcal$ and
$\Wcal^\varnothing$ is unimportant for the bulk flow, and
$\Wcal^\varnothing$ is denoted by $\Wcal$ in \cite{BBS-saw4-log}.

\subsection{Change of variables}
\label{sec:changevariables}

Theorem~\ref{thm:bulkflow} is stated in terms of the parameters $m^2,g_0$, rather than
the parameters $g,\nu$ of the weakly self-avoiding walk.
The following proposition,
proved in \cite[Proposition~\ref{log-prop:changevariables}(ii)]{BBS-saw4-log},
relates these sets of parameters via the functions $z_0^c,\nu_0^c$ of
Theorem~\ref{thm:bulkflow} and \eqref{e:gg0}.

\begin{prop} \label{prop:changevariables}
  Let $d=4$ and let $\delta_1>0$ be sufficiently small.
  There exists a function $[0,\delta_1)^2 \to [0,\delta)^2$, written
  $(g,\varepsilon) \mapsto (\tilde m^2(g,\varepsilon),\tilde g_0(g,\varepsilon))$,
  such that \eqref{e:gg0} holds with $\nu = \nu_c(g)+\varepsilon$,
  if $z_0 = z_0^c(\tilde m^2,\tilde g_0)$ and $\nu_0 = \nu_0^c(\tilde m^2,\tilde g_0)$.
  The functions $\tilde m, \tilde g_0$ are right-continuous as $\varepsilon \downarrow 0$,
  with $\tilde m^2(g,0) = 0$, and $\tilde m^2(g,\varepsilon) > 0$ if $\varepsilon > 0$.
\end{prop}

We also write
\begin{equation}
    \tilde z_0(g,\varepsilon) =
    z_0^c(\tilde m^2(g,\varepsilon),\tilde g_0(g,\varepsilon)),
    \quad\quad
    \tilde \nu_0(g,\varepsilon) =
    \nu_0^c(\tilde m^2(g,\varepsilon),\tilde g_0(g,\varepsilon)).
\end{equation}
The functions $\tilde z_0, \tilde \nu_0$ are right-continuous as $\varepsilon \downarrow 0$.
For the problem without observables, considered in \cite{BBS-saw4-log}, we
analysed the sequence $Z_j$ by choosing variables as follows.
First,
starting from $(g,\nu)$, Proposition~\ref{prop:changevariables}
gives us $(\mgen^2,\ggen_0)$, and then Theorem~\ref{thm:bulkflow} gives us
an initial condition $U_0=(\ggen_0,\tilde z_0,\tilde \nu_0)$ for which there
exists a global bulk flow of the renormalisation group map.
In the next section, we extend this to include observables.

\section{Observable flow}
\label{sec:obflow}

It follows
from Proposition~\ref{prop:Glims} and \refeq{GZN} that
\begin{equation}
\lbeq{Gil-bis}
    G_{g,\nu_c}(a,b)
    =
    \lim_{\varepsilon \downarrow 0}\lim_{N \to \infty}
    G_{N,g,\nu_c + \varepsilon}(a,b)
    =
    \lim_{\varepsilon \downarrow 0}
    \big((1+z_0)
    \lim_{N\to\infty}
    Z_{N;\sigma\bar\sigma}^0\big)
    ,
\end{equation}
provided the parameters $(m^2,g_0,\nu_0,z_0)$ implicit on the right-hand side
obey \eqref{e:gg0} with $\nu = \nu_c(g)+\varepsilon$.
To analyse \eqref{e:Gil-bis} via the renormalisation group flow, our
remaining task is to supplement the bulk flow of
Theorem~\ref{thm:bulkflow} with the flow of the observable coupling
constants $\lambda_{a,j},\lambda_{b,j},q_{a,j},q_{b,j}$ and of the
observable part $\pi_*K_j$ of $K_j$. In other words, we extend
Theorem~\ref{thm:bulkflow} to the case of nonzero $\sigma$.  This is
truly an extension, in the sense that the bulk flow needs no
modification because the equations for
$\lambda_{a,j},\lambda_{b,j},q_{a,j},q_{b,j},\pi_* K_j$ depend on but
do not appear in the flow of $(g_j,z_j,\nu_j,\pi_\varnothing K_j)$
which corresponds to $\sigma=0$, by \eqref{e:bulk}.  With the
estimates provided by \cite{BS-rg-step}, we will prove
Theorem~\ref{thm:wsaw4} using the kind of perturbative calculations
familiar in the physics literature, in a mathematically rigorous
manner.

\subsection{Perturbative flow of observables}
\label{sec:pt}

\begin{defn}
\label{def:jab}
Given $a,b \in \Lambda$, the \emph{coalescence scale} $j_{\pp\qq}$ is
defined by
\begin{equation}
   \label{e:Phi-def-jc}
   j_{\pp \qq} =
   \big\lfloor
   \log_{L} (2 |\pp - \qq|)
   \big\rfloor
    .
\end{equation}
\end{defn}

The coalescence scale is related to the finite-range property of the covariance
decomposition mentioned in Section~\ref{sec:pgi}, namely that $C_{j;x,y}=0$ if $|x-y|\ge \frac 12 L^j$.
Thus $j_{ab}$ is such that $C_{j_{ab};a,b}=0$, but $C_{j_{ab}+1;a,b}$ need not be
zero.
By definition, $L^{-2j_{ab}}$ is bounded above and below by multiples of $|a-b|^{-2}$,
in fact $L^{j_{ab}} \le 2|a-b|$.

In \cite{BBS-rg-pt}, the flow of the coupling constants in $V$ is
computed at a perturbative level.  The perturbative flow is without
control of errors uniformly in the volume, and we address the uniform
control below.  The perturbative flow is determined by a map $V =
(g,\nu,z,\lambda_a,\lambda_b,q_a,q_b) \mapsto V_\pt =
(g_\pt,\nu_\pt,z_\pt,\lambda_{a,\pt},\lambda_{b,\pt},q_{a,\pt},q_{b,\pt})$;
here we are only interested in $\lambda,q$.
The perturbative flow of $\lambda, q$ is reported in
\cite[\eqref{pt-e:lambdapt2}--\eqref{pt-e:qpt2}]{BBS-rg-pt} as the
scale-dependent map $V \mapsto (\lambda_{\pt},q_{\pt})$ given, for
$x=a,b$, by
\begin{align}
  \label{e:lampt}
  \lambda_{x,\pt}
  &
  =
  \begin{cases}
    (1 - \delta[\nu w^{(1)}])\lambda_{x}  & (j+1 < j_{\pp\qq})
    \\
    \lambda & (j+1 \ge j_{\pp\qq}),
  \end{cases}
  \\
  \label{e:qpt}
  q_{x,\pt}
  &
  =
  q_{x}
  +
  \lambda_a \lambda_b
  \, C_{j+1;\pp,\qq}
.
\end{align}
In \refeq{lampt}--\refeq{qpt}, $j$ refers to the scale of the initial
$V$, with $(\lambda_\pt,q_\pt)$ being scale-$(j+1)$ objects.  Also,
$w^{(1)} = w_j^{(1)} = \sum_{x\in \Lambda} \sum_{i=1}^j
C_{i;0,x}$, and
\begin{equation}
    \delta[\nu w^{(1)}]
    =
    \nu^+ w_{j+1}^{(1)}
    -
    \nu w_{j}^{(1)}
    \quad
    \text{with
    $\nu^+=\nu + 2g C_{j+1;0,0}$.}
\end{equation}
The coalescence scale $j_{\pp \qq}$
has the property that $q_{\pt} = 0$ if $q=0$ for $j \le j_{\pp\qq}$ because the factor $C_{j+1;a,b}$
on the right-hand side of \eqref{e:qpt} is zero when $j+1 \le j_{\pp\qq}$.
The considerations that lead to the stopping of the flow of $\lambda$ at the coalescence
scale in \refeq{lampt}
are discussed in \cite[Section~\ref{pt-sec:loc-specs}]{BBS-rg-pt}.

As discussed above \eqref{e:IcircKdq}, it is convenient to express the
renormalisation group map in terms of $\delta q$ rather than $q$.  For
this, we identify elements $V \in \Qcal^{(0)}$ with elements of
$\Qcal^{(1)}$ having $q_a=q_b=0$, and, when $V \in \Qcal^{(0)}$ we
write $\delta q_\pt$ instead of $q_\pt$.

\subsection{A single renormalisation group step}

Now we consider the renormalisation group map
\begin{equation}
  \label{e:Fch-obs}
  (V,K) \mapsto (V_{+}^{(1)}, K_{+}) = (\delta q_{+},V_{+}, K_{+}),
\end{equation}
which pertains not only to the bulk, but also to the observable
coupling constants as well as $\pi_*K= (\pi_{a}K,\pi_{b}K,\pi_{ab}K)$.
To state the estimates we require from \cite{BS-rg-step} for the map
\refeq{Fch-obs}, we recall \eqref{e:domRGbulk}, and define similarly
\begin{equation} \label{e:domRG}
\begin{aligned}
  \domRG_j(\sgen)
  &= \{(g,\nu,z,\lambda_a,\lambda_b) \in \R^3 \times \C^2
  :
    C_{\DV}^{-1} \ggen <  g < C_{\DV} \ggen,
    \; L^{2j}|\nu|,|z| \le C_\DV \ggen,
   \\
   & \qquad\qquad\qquad  |\lambda_a|,|\lambda_b| \le C_\DV
    \}
    \times B_{\Wcal_j}(\DVa\chigen_j\ggen^3).
\end{aligned}
\end{equation}
The first factor is the same as \cite[\eqref{step-e:DV1}]{BS-rg-step},
but restricted to real values.  Compared to $\domRG^\varnothing$ of
\eqref{e:domRGbulk}, the coupling constants $\lambda_a,\lambda_b$ are
included in $\domRG$ of \eqref{e:domRG}.  Also, the Banach spaces
$\Wcal_j = \Wcal_j(\sgen)$ now pertain to $K$ with components in
$\Ncal^a, \Ncal^b, \Ncal^{ab}$; these spaces are discussed in detail
in
\cite[Sections~\ref{step-sec:coordinates}--\ref{step-sec:norms}]{BS-rg-step}.
The domain $\domRG^\varnothing$ is obtained by projecting both factors
in the definition \eqref{e:domRG} by the appropriate definitions of
$\pi_\varnothing$ on $\Qcal^{(0)}$ and $\Wcal_j$ separately.

A $j$-dependent norm on $\Qcal^{(1)}$ is defined by
\begin{equation} \label{e:Vnormdef}
  \|V\|_{\Qcal} = \max\{|g|, L^{2j}|\nu_j|, |z_j|,
  \ell_j\ell_{\sigma,j}
  |\lambda_a|, \ell_j\ell_{\sigma,j}|\lambda_b|, \ell_{\sigma,j}^2 |q_a|
  , \ell_{\sigma,j}^2 |q_b|\}
\end{equation}
where
\begin{equation}
\label{e:ellsigdef-bis}
    \ell_j = \ell_0 L^{-j},\quad
    \ell_{\sigma ,j}
=
    2^{(j-j_{ab})_+}L^{(j\wedge j_{\pp\qq})} \ggen
    .
\end{equation}
The significance of the weights $\ell_j$, $\ell_{\sigma,j}$ is explained in
\cite[Remark~\ref{IE-rk:hsigmot}]{BS-rg-IE}; the constant
$\ell_0>0$ is determined in \cite[\eqref{IE-e:CLbd}]{BS-rg-IE}
and is of no direct importance here.

The following theorem concerns a single renormalisation group step
\eqref{e:RGmap}, with observables.  It is a reduced version of the
main result of \cite{BS-rg-step}, combining the relevant parts of
\cite[Theorems~\ref{step-thm:mr-R}--\ref{step-thm:mr}, \ref{step-thm:Kmcont}]{BS-rg-step}
into a single statement.  Such a result was used
in \cite[Theorem~\ref{log-thm:step-mr-fv}]{BBS-saw4-log},
but now observables are included in $V$ and $K$.
In fact, only the observable part of the statement is of interest
here---the bulk flow is independent and has already been analysed---but
it is convenient to state the theorem in its general form,
applying to both bulk and observables simultaneously.
The bounds on derivatives provided by \cite{BS-rg-step} are not stated
in Theorem~\ref{thm:step-mr-fv} as they are not needed here.

The map $V_+^{(1)} = (\delta q_+, V_+)$ is a perturbation of the map $V_\pt$
discussed in Section~\ref{sec:pt}, and it is convenient to describe it
in terms of the difference
\begin{equation}
\lbeq{Rplusdef}
  R_+(V,K) = V_+^{(1)}(V,K) - V_\pt(V).
\end{equation}
Thus $R_+$ is an element of $\Qcal^{(1)}$ with a component for each
of the seven coupling constants $(g,\nu,z,\lambda_a,\lambda_b,\delta
q_a,\delta q_b)$, and $\delta q$ is defined by $\delta q = \frac 12
(\delta q_a + \delta q_b)$.  As in \cite{BBS-saw4-log}, considerable
care is required to express the continuity of the maps $R_+,K_+$ in
the mass parameter $m^2$, and we define the intervals
\begin{align}
\lbeq{massint}
    \Iint_j &= \begin{cases}
    [0,\delta] & (j<N)
    \\
    [\delta L^{-2(N-1)},\delta] & (j=N),
    \end{cases}
\end{align}
and, for $\mgen^2 \in \Iint_j$,
\begin{align}
\lbeq{Itilint}
    \Igen_j &= \Igen_j(\mgen^2) =
    \begin{cases}
    [\frac 12 \mgen^2, 2 \mgen^2] \cap \Iint_j & (\mgen^2 \neq 0)
    \\
    [0,L^{-2(j-1)}] \cap \Iint_j & (\mgen^2 =0).
    \end{cases}
\end{align}
For the statement of the theorem,
we write $\sgen = (\mgen^2,\ggen)$ and $\sgen_+ = (\mgen^2,\ggen_+)$.
We assume $\half \ggen_+ \leq \ggen \leq 2\ggen_+$
and write $\chigen = \chi_j(\mgen^2)$.
We subsequently use the explicit choice $\sgen = s_j$ and $\sgen_+ = s_{j+1}$,
discussed in Section~\ref{sec:bulk},
and the choice of $\DVa$ mentioned below \eqref{e:domRGbulk}.
Then in particular $\chigen = \chi_j$.

\begin{theorem} \label{thm:step-mr-fv}
  Let $d =4$.
  Let $C_\DV$ and $L$ be sufficiently large.
  There exist $M>0$ and $\delta >0$ such that
  for $\ggen \in (0,\delta)$ and $\mgen^2 \in \Iint_+$, and with the domain
  $\domRG$ defined using any $\DVa> M$, the maps
  \begin{equation}
  \lbeq{RKplusmaps}
    R_+:\domRG(\sgen) \times \Igen_+(\mgen^2) \to \Qcal^{(1)},
    \quad
    K_+:\domRG(\sgen) \times \Igen_+(\mgen^2) \to \Wcal_{+}(\sgen_+)
  \end{equation}
  are analytic in $(V,K)$,
  and satisfy the estimates
  \begin{equation}
    \label{e:RKplus}
    \|R_+\|_{\Qcal}
    \le
    M\chigen\ggen^{3}_{+}
    , \qquad
    \|K_+\|_{\Wcal_+}
    \le
    M\chigen \ggen^{3}_{+}
    .
  \end{equation}
  In addition, $R_+,K_+$ are jointly continuous in all arguments $m^{2}, V,K$.
\end{theorem}

In a precise and non-trivial sense, Theorem~\ref{thm:step-mr-fv} shows
that the error to the perturbative calculation of Section~\ref{sec:pt}
is of third-order in the coupling constants.  However, unlike the bulk
coupling constants, which remain small, the observables are not small,
e.g., $\lambda_0=1$, and this is compensated by the weights in
\eqref{e:ellsigdef-bis}.

In the remainder of the paper, we write
\begin{equation}
  \text{$f\prec g$ when there is a $C>0$ such that $f \le Cg$;}
\end{equation}
the constant $C$ is always uniform in $g,\varepsilon$  and the
scale $j$ but may depend on $L$.

For $x=a,b$, let $R^{\lambda_x}_+$ denote the coupling constant
corresponding to $\lambda_{x}$ in $R_+$, and similarly for
$R^{q_x}_+$.  In \cite[Proposition~\ref{step-prop:Rcc}]{BS-rg-step},
it is shown that or $(V,K) \in \domRG_j$ and $x=a,b$,
\begin{align}
\lbeq{vlam}
   |R^{\lambda_{x}}_+|
&
\prec
    \chi_{j}  \ggen_{j}^{2} \1_{j< j_{\pp \qq}}
    ,
\\
\lbeq{vq}
    |R^{q_{x}}_+|
&
\prec
    |a-b|^{-2}
    \chi_{j}
   4^{-(j-j_{ab})} \ggen_{j}\1_{j \ge  j_{\pp \qq}}
.
\end{align}
The perturbative contribution $\lambda_{\pt,x}$ to the observable
coupling constant is independent of $x=a,b$, as is apparent from
\eqref{e:lampt}.  However, the paving of the torus $\Lambda$ by blocks
breaks translation invariance, and this allows $\lambda_x$ to have
non-perturbative contributions that depend on the relative positions
of $x=a,b$ within blocks.  Nevertheless, our main result
Theorem~\ref{thm:wsaw4} does not depend on the positions of $a,b$ in
the initial paving of $\Lambda$ by blocks.

\subsection{Observable flow}

The achievement of Theorem~\ref{thm:step-mr-fv} is to show that if
$(V_j,K_j)$ lies in the domain $\domRG_j$, then we have good control
of $\lambda_{x,j+1}, q_{x,j+1}$ and also the observable part of
$K_{j+1}$ (whose bulk part has been controlled along with the bulk
coupling constants already in Theorem~\ref{thm:bulkflow}).  The
following proposition links scales together via an inductive argument
to conclude that $(V_j,K_j)$ remains in $\domRG_j$ for \emph{all} $j
\le N$.

In particular, this requires that the bulk flow is well-defined for
all $j \leq N$.  For this, we recall that, given the parameters
$(m^2,g_0) \in [0,\delta)^2$, the critical initial conditions for the
global existence of the bulk renormalisation group flow are given by
\begin{equation} \lbeq{V0c}
  U_{0} = U_0^c = (g_0, z_0^c(m^2,g_0), \nu_0^c(m^2,g_0)),
\end{equation}
by Theorem~\ref{thm:bulkflow}.
We also recall the corresponding sequence $U_j(m^2,g_0)$.

According to \cite[\eqref{step-e:plusindep}]{BS-rg-step},
in the presence of observables, \refeq{bulk}
is supplemented by the statement that, for $x=a,b$,
\begin{equation}
\lbeq{plusindep}
\begin{aligned}
    &\text{if $\pi_x V=0$ and $\pi_x K(X)=0$ for all $X \in \Pcal$ then}
    \\
    &\text{$\pi_x V_+=\pi_{ab} V_+=0$ and $\pi_x K_+(U)
    =\pi_{ab} K_+(U)=0$ for all $U \in \Pcal_+$,}
\end{aligned}
\end{equation}
and, in addition, $\lambda_{a,+}$ is independent of each of
$\lambda_b$, $\pi_b K$, and $\pi_{ab}K$, and the same is true with
$a,b$ interchanged.

As a consequence, using Theorem~\ref{thm:step-mr-fv}, the next
proposition shows that the flow with observables, and with initial
conditions
\begin{equation} \label{e:V0ob}
  \pi_\varnothing V_0 = U_0^c,
  \quad \lambda_{x,0} \in \{0,1\}
  , \quad q_{x,0} = 0, \quad (x=a,b)
\end{equation}
exists for all $j \leq N$.  Note that we permit one or both of
$\lambda_{x,0}$ to equal zero, and in this case we regard the
observable at $x$ as being absent, so the concept of coalescence
becomes vacuous.  We therefore use the convention that
\begin{equation}
    j_{ab} = \infty
    \quad \text{if $\lambda_{a,0}=0$ or $\lambda_{b,0}=0$}.
\end{equation}

\begin{prop}
\label{prop:qlamN}
Let $\lambda_{x,0}\in \{0,1\}$ and $q_{x,0}=0$ for $x=a,b$.
\\
(i) For $(m^2,g_0) \in [\delta L^{-2(N-1)},\delta) \times (0,\delta)$,
there is a choice of $(q_{a,j},q_{b,j},V_j,K_j)$ such that
\refeq{Kspace-objective} holds for $0 \le j \le N$.
This choice is such that $\pi_\varnothing V_j = U_j(m^2,g_0)$.  If
$\lambda_{x,0}=0$ then $\lambda_{x,j}=0$ for all $0 \le j \le N$,
whereas if $\lambda_{x,0}=1$ then
\begin{align}
\lbeq{lam}
    \lambda_{x,j}
    &=
    \begin{cases}
    (1+\nu_jw_j^{(1)})^{-1}
    \Big(1 + \sum_{k=0}^{j-1} \vch_{\lambda_{x},k}\Big)
    & (j+1<j_{ab})
    \\
    \lambda_{j_{ab}-1} & (j+1 \ge j_{ab})
    .
    \end{cases}
\end{align}
If $\lambda_{x,0}=1$ for one or both of $x=a,b$ then
$q_{a,j}=q_{b,j}=0$ for all $0 \le j \le N$, whereas if
$\lambda_{a,0}=\lambda_{b,0}=1$ then, for $x=a,b$,
\begin{align}
\lbeq{q}
    q_{x,j}
    &=
    \sum_{i=j_{ab}-1}^{j-1}
    \left(
    \lambda_{a,j_{ab}-1}\lambda_{b,j_{ab}-1}
    \,C_{i+1;\pp,\qq}
    +
    v_{q_{x},i}
    \right)
    ,
    \\
\lbeq{K}
    \|K_j\|_{\Wcal_j}
    &\le
    M
    \chi_j \ggen_j^3.
\end{align}
In the above estimates, $M$ is the constant appearing in
\refeq{RKplus}, and $\vch_{\lambda_x,j},v_{q_x,j} \in \C$
obey, uniformly in $(m^2,g_0) \in [0,\delta)^2$,
\begin{equation}
\lbeq{vbds}
    |\vch_{\lambda_{x},j}| \prec \chi_j \ggen_j^2 \1_{j < j_{ab}},
    \quad\quad
    |v_{q_{x},j}| \prec
    |a-b|^{-2}
    \chi_j
   4^{-(j-j_{ab})} \ggen_{j} \1_{j \ge j_{\pp \qq}}.
\end{equation}
(ii)
  For $j \le N$,
  each of $\lambda_{x,j}, \delta q_{x,j}, q_{x,j}$ is independent of $N$,
  in the sense that, e.g., $q_{x,1},\ldots,q_{x,N}$ have the same values on
  $\Lambda_N$ as on a larger torus $\Lambda_{N'}$ with $N'>N$.
  In addition, each
  is defined as a continuous function of $(m^2,g_0) \in [0,\delta)^2$.
  Finally, for $j<j_{ab}$,
  $\lambda_{a,j}$ is independent of $\lambda_{b,0}$,
  and $\lambda_{b,j}$ is independent of $\lambda_{a,0}$.
\end{prop}

\begin{proof}
To simplify the notation, we drop the labels $x=a,b$ from $\lambda,q$
when their role is insignificant.

\smallskip\noindent
(i)
As a preliminary step, we introduce a change to variables that
diagonalise the evolution of $\lambda$ to linear order in $V$.  For
$(V_j,K_j)$, we write $\lambda_\pt = \lambda_\pt(V_j)$ and
$v_{\lambda,j}=R^\lambda_{j+1}(V_j,K_j)$.  Then the
$\lambda$-component of \eqref{e:Fch-obs} can be written as
$\lambda_{j+1}=\lambda_\pt + v_{\lambda,j}$.  We define
\begin{equation}
  \lambdach_j = \lambda_j(1+\nu_jw_j^{(1)}).
\end{equation}
By \eqref{e:lampt}, the recursion for $\lambdach_j$ can then be written as
\begin{equation}
\lbeq{lambdachrec}
    \lambdach_{j+1} = \lambdach_j + \vch_{\lambda,j}
\end{equation}
with
\begin{equation}
\lbeq{vchdef}
    \vch_{\lambda,j}
    =
    (\nu_{j+1}-\nu_j^+)\lambda_j w_{j+1}^{(1)} +
    v_{\lambda ,j}(1+ \nu_{j+1} w_{j+1}^{(1)})
  - \delta_j[\nu w^{(1)}]\lambda_j  \nu_{j+1} w_{j+1}^{(1)}
  .
\end{equation}
The solution to \refeq{lambdachrec} with initial condition $\lambda_0=1$
is
$\lambdach_j = 1 + \sum_{k=0}^{j-1}\vch_{\lambda,k}$,
and hence
\begin{align}
    \label{e:lamrec}
    \lambda_{j}
    &
    =
    (1+\nu_jw_j^{(1)})^{-1}
    \Big(1 + \sum_{k=0}^{j-1} \vch_{\lambda,k}\Big)
    .
\end{align}
By \eqref{e:qpt} and \refeq{Rplusdef}, and with $v_{q,j} =
R^q_{j+1}(V_j,K_j)$, $\delta q_j$ is simply given by
\begin{equation} \label{e:obs-flow-q}
  \delta q_{j+1} = \delta q_\pt + v_{q,j} =
  \lambda_{a,j}\lambda_{b,j}
  C_{j+1;a,b} + v_{q,j}.
\end{equation}

Now we can prove \refeq{lam}--\refeq{vbds} by induction on $j$, with
induction hypothesis:
\begin{align*}
    {\rm IH}_j:  \hspace{3mm} &
    \text{for all $k \le j$,
    $(V_k,K_k) \in \domRG_k$,
    \refeq{lam}--\refeq{vbds} hold with $j$ replaced by $k$.}
\end{align*}
By direct verification, ${\rm IH}_0$ holds
(with $\vch_{\lambda,-1}=v_{q,-1}=0$).

We assume ${\rm IH}_j$ and show that it implies ${\rm IH}_{j+1}$. By
${\rm IH}_j$ and the bound \refeq{RKplus} of
Theorem~\ref{thm:step-mr-fv}, $K_{j+1}$ obeys \refeq{K}.  In
particular, this estimate implies $K_{j+1} \in B_{\Wcal_j}(\DVa \chi_j
\ggen_j)$.

By \eqref{e:bulk}--\eqref{e:RGmapbulk}, $\pi_\varnothing V_{j} =
U_{j}$ for all $j$, and by Theorem~\ref{thm:bulkflow}, $U$ satisfies
the bounds required for $\pi_\varnothing V$ in the definiton of
$\domRG$. Therefore, to verify $(V_{j+1},K_{j+1}) \in \domRG_{j+1}$,
it suffices to show $|\lambda_{j+1}| \leq C_\Dcal$.

By \refeq{Rplusdef}, \refeq{lampt}, and \refeq{vlam}, $\lambda_j =
\lambda_{j_{ab}-1}$ for all $j \ge j_{ab}$, so we assume that
$j<j_{ab}-1$.  To estimate $\vch_{\lambda,j}$, we use the fact that
$|\lambda_j| \le C_\DV$ by assumption, and $|\nu_j| \prec
L^{-2j}\chi_j \ggen_j$ by Theorem~\ref{thm:bulkflow}.  We apply
\cite[Lemma~\ref{pt-lem:wlims}]{BBS-rg-pt} and
\cite[\eqref{step-e:factnu}]{BS-rg-step} to see that $w_j^{(1)} \prec
L^{2j}$ and $|\nu_{j+1}-\nu_j^+| \prec L^{-2j}\chi_j \ggen_j^2$, and
also $ |\nu_j|w_j^{(1)} \prec L^{-2j} \chi_j \ggen_j w_j^{(1)} \prec
\chi_j \ggen_j$.  The factor $v_{\lambda,j}$ is bounded via
\refeq{vlam}, and the last term on the right-hand side of
\refeq{vchdef} is similarly bounded (without any need for cancellation
in the $\delta$ term).  We conclude that
$|\vch_{\lambda,j}|\prec\chi_j\ggen_{j}^{2}$, as required.  With
\refeq{chisum}, this leads to $|\lambda_{j+1}| = 1+O(g_0) \le C_{\DV}$
(since we have assumed above \refeq{domRGbulk} that $C_{\DV}\ge 2$).
This establishes that $\lambda_{j+1}$ obeys the condition required in
the definition of the domain $\domRG_{j+1}$, and all necessary
properties for $\lambda_{j+1}$ have been established.

By \eqref{e:obs-flow-q} and \eqref{e:vq} with ${\rm IH}_j$, \refeq{q}
holds with $v_{q,j}$ obeying \refeq{vbds}.  This advances the
induction and completes the proof of part~(i).

\smallskip \noindent
(ii)
The $N$-independence of $\lambda_j,\delta q_j$ follows exactly as in
the proof of \cite[Proposition~\ref{log-prop:KjNbd}]{BBS-saw4-log}, so
we only sketch the argument.  By \eqref{e:lampt}--\eqref{e:qpt},
$\lambda_\pt$ and $\delta q_\pt$ are independent of $N$.  Moreover, by
\cite[Proposition~\ref{step-prop:VZd}(i)]{BS-rg-step}, $R_+(V,K)$ is
independent of $N$ provided that $V$ is independent of $N$ and that
the family $K$ has Property~$\Z^d$ defined in \cite{BS-rg-step}.  That
the renormalisation group map preserves Property~$\Z^d$ for $K$ is
shown in \cite[Proposition~\ref{step-prop:KplusZd}]{BS-rg-step}.

To show that $\lambda_j, \delta q_j$ (and thus also $q_j$) are
continuous as functions of $(m^2,g_0) \in [0,\delta)^2$, assuming that
$V_0 = V_0^c(m^2,g_0)$, we can proceed exactly as in
\cite[Section~\ref{log-sec:Fcont}]{BBS-saw4-log}.  The definition of
\emph{continuous functions of the renormalisation group coordinates at
scale-$j$}, provided by
\cite[Definition~\ref{log-def:flowcont}]{BBS-saw4-log} for the bulk
coordinates, applies literally also to the renormalisation group
coordinates with observables. By Theorem~\ref{thm:step-mr-fv} for
$R_+$ and \cite[Lemma~\ref{pt-lem:wlims}]{BBS-rg-pt} for $V_\pt$, both
of $\lambda_+,\delta q_+$ are continuous functions of the
renormalisation group coordinates at scale-$j$.  By
\cite[Proposition~\ref{log-prop:Fcont}]{BBS-saw4-log}, which also
applies literally with observables, we conclude continuity of
$\lambda_j,\delta q_j$ for all $j$.

Finally, it follows inductively from \refeq{plusindep} and the
statement below \refeq{plusindep} that if $j<j_{ab}$ then
$\lambda_{a,j}$ is independent of $\lambda_{b,0}$, and vice versa, as
required.  This completes the proof.
\end{proof}

In the following lemma, we denote the derivative of
$Z_N^0(\phi,\phib)$ with respect to $\phib$, in the direction of a
test function $J:\Lambda \to \C$, as $D_{\phib} Z_N^0(\phi,\phib;J)=
\frac{d}{dt}Z_N^0(\phi,\phib+tJ)|_0$.  Let $1$ denote the constant
test function $1_x=1$ for all $x \in \Lambda$.  We systematically use
subscripts $\sigma$ or $\sigma\sigmab$ to denote the coefficient of
$\sigma$ or $\sigma\sigmab$ in $F\in \Ncal$, under the decomposition
\refeq{Ncaldecomp}.  For example, we write
$K_{N;\sigma\sigmab}(\Lambda)=
\frac{1}{\sigma\sigmab}\pi_{ab}K_N(\Lambda)$.

\begin{lemma}
\label{lem:lamid}
Let $\lambda_{x,0}\in \{0,1\}$ and $q_{x,0}=0$ for $x=a,b$.
The flow of Proposition~\ref{prop:qlamN} obeys
\begin{equation}
    \lambda_{a,N} =
    D_{\phib} Z_{N;\sigma}^{0}(0,0; 1)
    - D_{\phib}W_{N;\sigma}^{0}(\Lambda; 0,0; 1)
    - D_{\phib}K_{N;\sigma}^{0}(\Lambda; 0,0; 1).
\end{equation}
\end{lemma}

\begin{proof}
As in \cite[\eqref{log-e:chibarm-bisIK}]{BBS-saw4-log},
\begin{equation}
  Z_N^0 = I_N^0(\Lambda) + K_N^0(\Lambda)
  = e^{-V_N^0(\Lambda)} (1+W_N^0(\Lambda)) + K_N^0(\Lambda).
\end{equation}
Therefore, since $\pi_a(FG) = (\pi_aF)(\pi_\varnothing G) + (\pi_\varnothing F)(\pi_a G)$,
and since
\begin{equation}
  \pi_\varnothing (e^{-V_N^0(\Lambda)}) = e^{-U_N^0(\Lambda)},
  \quad \pi_a (e^{-V_N^0(\Lambda)}) = \sigma\lambda_{a,N} \phib_a,
\end{equation}
we obtain
\begin{align}
  Z_{N;\sigma}^0
  &= \lambda_{a,N}\phib_a e^{-U_N^0(\Lambda)} (1+W_N^{0,\varnothing}(\Lambda))
  + e^{-U_N^0(\Lambda)}  W_{N;\sigma}^{0}(\Lambda) + K_{N;\sigma}^{0}(\Lambda).
\end{align}
Differentiating with respect to $\phib$ at $(\phi,\phib)=(0,0)$, we
obtain
\begin{equation}
    D_{\phib} Z_{N;\sigma}^{0}(0,0; 1) =
    \lambda_{a,N}
    +D_{\phib}W_{N;\sigma}^{0}(\Lambda; 0,0; 1)
    + D_{\phib}K_{N;\sigma}^{0}(\Lambda; 0,0; 1),
\end{equation}
where we used $e^{-U_N^0(\Lambda; 0,0)} = 1$ and the fact that
$W_N^{0,\varnothing}(\Lambda; 0,0) = 0$ since $W_N^{0,\varnothing}$ is
a polynomial in $\phi$ with no monomials of degree below two.
\end{proof}

The $W$ and $K$ terms in the statement of Lemma~\ref{lem:lamid} are
estimated using the following lemma.

\begin{lemma}
\label{lem:WK} The following bounds hold uniformly in $m^2 \in [\delta
L^{-2(N-1)},\delta)$:
\begin{align}
\lbeq{Kg1}
    \left|
    K_{N;\sigma\bar\sigma}^0(\Lambda; 0,0)
    \right|
    &\prec
    \frac{1}{4^{(N-j_{ab})_+}}\frac{1}{|a-b|^{2}}
    \chi_N \gbar_N,
\\
\lbeq{DKbd}
  |D_{\phib}K_{N;\sigma}^{0}(\Lambda;0,0;1)|
  &\prec \chi_N \gbar_N^2
  \left(\frac{L}{2}\right)^{(N-j_{ab})_+},
\\
\label{e:WNbdz}
  |D_{\phib} W_{N;\sigma}^0(\Lambda; 0, 0;  1)|
  &\prec \chi_N \bar g_N
  \left(\frac{L}{2}\right)^{(N-j_{ab})_+}
  .
\end{align}
\end{lemma}

\begin{proof}
Recall the definitions of the $\Qcal$ norm from \refeq{Vnormdef}, and
the definitions of the $T_{0,j}(\ell_j)$ and $\Phi_j(\ell_j)$ norms
from \cite[Section~\ref{log-sec:flow-norms}]{BBS-saw4-log}.  By
definition of the $T_{0,j}(\ell_j)$ norm,
\begin{equation}
    |K_{N;\sigma\bar\sigma}^0(\Lambda; 0,0)| \le \|K_N(\Lambda)\|_{T_{0,N}(\ell_N)},
\end{equation}
and,
for any $F\in \Ncal^\varnothing$ and any test function $J:\Lambda \to \C$,
\begin{equation}
\lbeq{DphibF}
  |D_{\phib} F^0(0, 0; J)|
  \leq   \|F\|_{T_{0,N}(\ell_N)} \|J\|_{\Phi_N(\ell_N)}.
\end{equation}

Recall from \cite[\eqref{loc-e:Fnormsum}]{BS-rg-loc} that in the $T_0$ norm
each occurrence of $\sigma$ or $\sigmab$ gives rise to a weight
\begin{equation}
\label{e:ellsigdef}
    \ell_{\sigma ,j}
=
    2^{(j-j_{ab})_+}L^{(j\wedge j_{\pp\qq})} \ggen_{j}
.
\end{equation}
There is therefore a factor $\ell_{\sigma,j}^2$ inside the norm of
$\pi_{ab}K_j$.  We apply \cite[\eqref{step-e:KWbd}]{BS-rg-step}, which
uses this fact, and which implies that the bound
\begin{equation}
\lbeq{KWbd}
    |K_{N;\sigma\sigmab}^0(0,0)|
    \le
    \ell_{\sigma ,j}^{-2}
    \|K_N(\Lambda)\|_{T_{0,N}(\ell_N)}
    \le
    \ell_{\sigma ,j}^{-2}
    \|K_N\|_{\Wcal_N}
    \prec
    4^{-(N-j_{ab})_+}L^{-2j_{\pp\qq}} \chi_N \ggen_{N}
\end{equation}
holds uniformly in $m^2 \in [\delta L^{-2(N-1)},\delta)$.  As
mentioned below Definition~\ref{def:jab}, $L^{2j_{ab}}$ and $|a-b|$
are comparable.  With \refeq{K} and \refeq{gbarggen}, this shows that
that \refeq{Kg1} holds.

By definition, $\|1\|_{\Phi_N(\ell_N)}=\ell_N^{-1}$ (see
\cite[\eqref{log-e:1norm}]{BBS-saw4-log}).  With \refeq{DphibF}, this
gives
\begin{equation}
\lbeq{DphibK}
\begin{aligned}
  |D_{\phib}K_{N;\sigma}^{0}(\Lambda;0,0;1)|
  &\leq
  \ell_{\sigma,N}^{-1} \|K_N(\Lambda)\|_{T_{0,N}(\ell_N)} \|1\|_{\Phi_N(\ell_N)}
  =
  \ell_{\sigma,N}^{-1}\ell_N^{-1} \|K_N\|_{\Wcal_N}
  .
\end{aligned}
\end{equation}
With \refeq{ellsigdef-bis} and \refeq{K}, this proves \refeq{DKbd}.
Finally, by \cite[Proposition~\ref{IE-prop:Wnorms}]{BS-rg-IE},
\begin{equation} \label{e:Wbilinbd}
  \|W_{N}(\Lambda)\|_{T_{0,N}} \prec \chi_N g_N^2,
\end{equation}
and \refeq{WNbdz} then follows as in \refeq{DphibK}.
\end{proof}

The next two lemmas apply Proposition~\ref{prop:qlamN} to study limits
of the sequences $\lambda_{x,j},q_{x,j}$.  By
Proposition~\ref{prop:qlamN}(ii), $q_{x,j}$
is independent of $N$ (assuming that $N$ is larger than $j_{\pp\qq}$),
and $\lambda_{x,j}$ is independent of $j_{\pp\qq}$ and $N$ if
$j<j_{ab}\leq N$, and we can therefore define sequences
$\lambda_{x,j}^*$ for all $j \in \N_0$, with $\lambda^*_{x,0}=1$, such
that $\lambda_{x,j} = \lambda_{x,j}^*$ for $j < j_{ab}$.  The sequence
$\lambda_{a,j}^*$ is independent of $\lambda_{b,0}$, and vice versa.
By definition,
\begin{equation}
    \lambda_{x,j} = \lambda^*_{x,j\wedge (j_{ab}-1)},
\end{equation}
and
\begin{equation}
    \lambda_{a,j} = \lambda_{a,j}^* \quad \text{for all $j \le N$ when $\lambda_{b,0}=0$.}
\end{equation}
We make the dependence on $(m^2,g_0)$ explicit by writing
$\lambda_{x,j} = \lambda_{x,j}(m^2,g_0)$ and $q_{x,j} =
q_{x,j}(m^2,g_0)$.

\begin{lemma}
\label{lem:lamlim}
For $(m^2,g_0) \in [0,\delta)^2$, for $x=a$ or $x=b$, for
$\lambda_{x,0}=1$, and for $j \in \N_0$,
\begin{equation}
    |1-\lambda^*_j (m^2,g_0)| \prec  \chi_{j}\gbar_{j}.
\end{equation}
In particular,
\begin{equation}
    |1 - \lambda_{x,j_{ab}-1}(m^2,g_0)|
    \prec \chi_{j_{ab}}\gbar_{j_{ab}}.
\end{equation}
\end{lemma}

\begin{proof}
By Proposition~\ref{prop:qlamN}, the flow of $\lambda_a$ is
independent of the choice of $\lambda_{b,0}$, and vice versa.  We give
the proof for the case $x=a$, and the same argument applies to $x=b$.

We choose the initial conditions
$(\lambda_{a,0},\lambda_{b,0})=(1,0)$.  As discussed above
Proposition~\ref{prop:qlamN}, in this case we have $j_{ab}=N$.  By
Lemma~\ref{lem:lamid},
\begin{equation}
\lbeq{lamD}
    \lambda_{a,N}^* =
    D_{\phib} Z_{N;\sigma}^{0}(0,0; 1)
    - D_{\phib}W_{N;\sigma}^{0}(\Lambda; 0,0; 1)
    - D_{\phib}K_{N;\sigma}^{0}(\Lambda; 0,0; 1).
\end{equation}
By Lemma~\ref{lem:WK}, this gives
\begin{equation}
\lbeq{lamDbd}
    \lambda_{a,N}^* =
    D_{\phib} Z_{N;\sigma}^{0}(0,0; 1)
    +O(\chi_N g_N).
\end{equation}
The limit of the first term on the right-hand side of \refeq{lamDbd},
as $N\to\infty$, can be evaluated exactly, as follows.  Let $C$ be the
covariance defined in \refeq{Cdef}.  Recall from
\cite[\eqref{log-e:GamZ0}]{BBS-saw4-log} that, for any external field
$J: \Lambda \to \C$,
\begin{equation} \label{e:SigmaaZN0a}
  \Sigma_{a}(J,\bar J) = \Ex_C\left(e^{-V_{0}(\Lambda)+(J,\phib)+(\bar J, \phi)}\right)
  = e^{(J,C\bar J)} Z_N^0(CJ, C\bar J),
\end{equation}
where the superscript $0$ denotes projection onto the degree-$0$ part
of the form $Z_N$.  As opposed to \cite{BBS-saw4-log}, we include the
observable term $\sigma\phib_a$ in $V_0$ and $Z_N$ here, and we
emphasise this by writing $\Sigma_{a}$ instead of $\Sigma$; the
potential $V_0$ without observable terms is again denoted by $U_0$.
Each side of \refeq{SigmaaZN0a} has a decomposition as in
\refeq{Ncaldecomp}, and we equate the coefficients of $\sigma$ in the
components in $\Ncal^a$ to obtain
\begin{equation} \label{e:EZN0a}
  \Ex_C\left(e^{-U_{0}(\Lambda)+(J,\phib)+(\bar J, \phi)} \phib_a\right)
  = e^{(J,C\bar J)}
  Z_{N;\sigma}^0
  (CJ,C\bar J).
\end{equation}

Let $1$ be the constant test function $1_x=1$ for all $x \in \Lambda$.
Then $C1 = m^{-2}1$.
Differentiation of \refeq{EZN0a} at $(0,0)$ with respect to $J$, in
direction $1$, gives
\begin{align}\label{e:sumEZN0a}
  \sum_x \Ex_C\left(e^{-U_{0}(\Lambda)} \phi_x \phib_a\right)
  = D_{\bar \phi} Z_N^{0;a}(0,0; C1)
  = m^{-2} D_{\bar \phi} Z_{N;\sigma}^{0}(0,0; 1)
  .
\end{align}
By translation invariance of $\Ex_C$ and $U_0$, the left-hand side is
independent of $a\in\Lambda$.  In fact, it is equal to $\hat\chi_N$
defined in \cite[\eqref{log-e:chiNhatdef}]{BBS-saw4-log}, which
converges to $m^{-2}$ as $N\to\infty$, by
\cite[Theorem~\ref{log-thm:suscept-diff}]{BBS-saw4-log}.  Therefore,
\begin{equation} \label{e:ZN0ato0}
  \lim_{N\to\infty} D_{\phib} Z_{N;\sigma}^0(0,0; 1) = 1.
\end{equation}
We then apply Lemma~\ref{lem:WK} (with $(N-j_{ab})_+=0$), together
with $\chi_N g_N \to 0$, to conclude from \eqref{e:ZN0ato0} that the
right-hand side of \refeq{lamDbd} tends to $1$.  On the other hand, by
\refeq{lam}, together with \refeq{vbds} and the estimate $|\nu_j|
w_j^{(1)} \prec \chi_j \ggen_j$ used in the proof of
Proposition~\ref{prop:qlamN},
\begin{align} \label{e:lamlim}
    \lim_{N\to\infty } \lambda_{a,N}^*
    =
    1 + \sum_{k=0}^{\infty}\vch_{\lambda_x,k},
    \quad\quad
    \text{so} \quad \sum_{k=0}^{\infty}\vch_{\lambda_x,k}=0.
\end{align}
(Note that the convergence of the sum in \refeq{lamlim} is guaranteed
by \refeq{vbds} and \refeq{chisum}.)  Finally, by \refeq{lam} and
\refeq{chisum}, uniformly in $(m^2,g_0)$ we have
\begin{align}
\label{e:lambda3}
    \lambda_{j}^* - 1
&=
    - \nu_j w_j^{(1)} \lambda_{j} -
    \sum_{k= j}^\infty \vch_{\lambda,k}
    = O(\chi_j\gbar_j)
,
\end{align}
and the proof is complete.
\end{proof}

\begin{lemma}
\label{lem:lamqflow}
For $(m^2,g_0) \in [0,\delta)^2$ and $x=a,b$, the limit
\begin{align}
    q_{x,\infty}(m^2,g_0)  = \lim_{j \to \infty}q_{x,j}(m^2,g_0),
\label{e:qlims}
\end{align}
exists, is continuous, and, as $|a-b|\to\infty$,
\begin{equation}
\lbeq{qLap}
  q_{x,\infty}(0, g_0)
  =
  (-\Delta_{\Z^4})^{-1}_{ab}
  \left( 1 + O \left( \frac{1}{\log |a-b|} \right) \right)
  .
\end{equation}
\end{lemma}

\begin{proof}
We again drop the labels $x=a,b$ from $\lambda,q$ when their role is
insignificant.

By \refeq{q},
\begin{equation}
    \label{e:qN}
    q_{j}
    =
    \sum_{i=j_{\pp\qq}-1}^{j-1}
    \left(
    \lambda_{a,j_{ab}-1}\lambda_{b,j_{ab}-1}
    \,C_{i+1;\pp,\qq}
    +
    v_{q,i}
    \right).
\end{equation}
Since $C_{i+1;a,b}=0$ for $i<j_{\pp\qq}$, we can restore the scales
$i< j_{\pp\qq}$ to the sum in the first term on the right-hand side.
In the limit $j \to \infty$, we obtain the complete finite-range
decomposition for the inverse Laplacian on $\Z^4$ as in
\eqref{e:Zddecomp},
\begin{equation}
    \sum_{i=j_{\pp\qq}-1}^{\infty}
    C_{i+1;\pp,\qq}
    =
    \sum_{i=0}^{\infty}
    C_{i+1;\pp,\qq}
    =
    (-\Delta_{\Z^4}+m^2)^{-1}_{ab}.
\end{equation}
The dependence of $\lambda_{x,j_{ab}-1}$ is continuous in
$[0,\delta)^2$ by Proposition~\ref{prop:qlamN}, and
$(-\Delta_{\Z^4}+m^2)^{-1}_{ab}$ is continuous in $m^2 \in
[0,\delta)$.
By Proposition~\ref{prop:qlamN}, in the limit $j \to \infty$ the sum
of $v_{q,i}$ on the right-hand side of \refeq{qN} is a uniformly
convergent sum of terms that are continuous.  Therefore the sum
$q_\infty(m^2,g_0)$ is also continuous, and
\begin{equation}
    \label{e:q-infty}
    q_\infty
    =
    \lambda_{a,j_{ab}-1}\lambda_{b,j_{ab}-1}
    (-\Delta_{\Z^4}+m^2)^{-1}_{ab}
    +
    \sum_{i=j_{\pp\qq}}^{\infty}
    v_{q,i}
    .
\end{equation}

By \refeq{vq},
\begin{equation}
    \sum_{i=j_{\pp\qq}}^{\infty} | v_{q,i}|
    \prec
    |a-b|^{-2}
    \sum_{i=j_{\pp\qq}}^{\infty}
    4^{-(i-j_{ab})}
    \chi_i
    \ggen_{i}
    \prec
    |a-b|^{-2}
    \chi_{j_{ab}} \gbar_{j_{ab}}
    .
\end{equation}
Therefore,
\begin{align}
\lbeq{qdiff}
    \big|
    q_\infty -
    \lambda_{a,j_{ab}-1}\lambda_{b,j_{ab}-1}
    (-\Delta_{\Z^4}+m^2)^{-1}_{ab}
    \big|
    &\prec
    \chi_{j_{\pp \qq}} \gbar_{j_{\pp \qq}}
    |\pp - \qq|^{-2}
    .
\end{align}
By Lemma~\ref{lem:lamlim},
\begin{equation}
\lbeq{lamlim2}
    |
    1
    -
    \lambda_{a,j_{ab}-1}\lambda_{b,j_{ab}-1}
    |
    \prec \chi_{j_{ab}}\gbar_{j_{ab}}.
\end{equation}
With \refeq{qdiff} and $|(-\Delta_{\Z^4}+m^2)^{-1}_{ab}| \prec
|a-b|^{-2}$, this gives
\begin{align}
    \big|
    q_\infty -
    (-\Delta_{\Z^4}+m^2)^{-1}_{ab}
    \big|
    &\prec
    \chi_{j_{\pp \qq}} \gbar_{j_{\pp \qq}}
    |\pp - \qq|^{-2}
    ,
\end{align}
uniformly in $(m^2,g_0)$.  By \refeq{chiglim},
$\chi_{j_{\pp\qq}}\gbar_{j_{\pp \qq}} \prec j_{ab}^{-1} \prec
(\log|a-b|)^{-1}$.  In particular, the limit $q_\infty (0,g_0)$ obeys
\refeq{qLap}.
\end{proof}

\subsection{Proof of main result}
\label{sec:alt-crit}

We now prove Theorem~\ref{thm:wsaw4}.  In addition to the study of
$q_j$, which provides the leading contribution, this requires the
estimate \refeq{Kg1} on $\pi_{ab}K_N$.

\begin{proof}[Proof of Theorem~\ref{thm:wsaw4}] For small
$g,\varepsilon>0$, set $\nu=\nu_c(g)+\varepsilon$, and let
$(m^2,g_0,\nu_0,z_0) = (\tilde m^2,\tilde g_0, \tilde \nu_0, \tilde
z_0)$ be the functions of $(g,\varepsilon)$ given by
Proposition~\ref{prop:changevariables}.  Since $z_0 = \tilde
z_0(g,\varepsilon) \to \tilde z_0(g,0)$ as $\varepsilon \downarrow 0$,
\refeq{Gil-bis} gives
\begin{equation}
\lbeq{Gil}
    G_{g,\nu_c}(a,b)
    =
    \big(1+\tilde z_0(g,0) \big)
    \lim_{\varepsilon \downarrow 0}
    \lim_{N\to\infty}
    Z_{N;\sigma\bar\sigma}^0
    (0)
    .
\end{equation}
The arguments $0$ on the right-hand side mean that the fields $\phi
,\psi$ are to be set to zero in $I_{N},K_{N}$.  Thus
$I_N^0(\Lambda)=1$, and for $K_N^0(\Lambda)$ only dependence on
$\sigma,\sigmab$ remains.  From \eqref{e:ZNcirc} we obtain
\begin{equation}
\label{e:ZNIK}
    Z_{N}^0(0)
    =
    e^{q_N\sigma\sigmab}
     (I_{N}^0(\Lambda,0) + K_{N}^0(\Lambda,0))
    =
    e^{q_N\sigma\sigmab}
    (1+
     K_{N}^0(\Lambda,0)  )
    ,
\end{equation}
with $q_N = \frac 12 (q_{a,N} + q_{b,N})$ as in \refeq{delqdef}.
Equating the coefficients of $\sigma\bar\sigma$ on both sides gives
\begin{align}
    Z_{N;\sigma\bar\sigma}^0(0)
    =
    q_{N}\big( 1 + \pi_\varnothing K_N^0(\Lambda,0) \big)
    +
    K_{N;\sigma\bar\sigma}^0(\Lambda,0)
    .
\end{align}
Since $\varepsilon > 0$ by assumption, it follows that $m^2>0$,
by Proposition~\ref{prop:changevariables}.
Therefore, for $N$ sufficiently large, the bounds \refeq{Kgnull} and \refeq{Kg1} hold.
In particular, by \eqref{e:chiglim},
\begin{equation}
\label{e:KNsslim}
    \lim_{N \to \infty}
    \pi_\varnothing K_N^0(\Lambda,0)=0,
    \quad\quad
    \lim_{N\to\infty}
    K_{N;\sigma\bar\sigma}^0(\Lambda,0)
    = 0
    ,
\end{equation}
and therefore
\begin{align}
\label{e:final}
    \lim_{N\to\infty}
    Z_{N;\sigma\bar\sigma}^0(0)
    &=
    \lim_{N\to\infty}q_N
    =
    q_\infty
    = \frac 12 (q_{a,\infty} + q_{b,\infty})
    .
\end{align}
With \eqref{e:Gil} and Lemma~\ref{lem:lamqflow}, this gives
\begin{align}
\label{e:final1}
    G_{g,\nu_c}(a,b)
    &=
    (1+\tilde z_0(g,0))
    \lim_{\varepsilon \downarrow 0}
    q_\infty
\nnb
    &=
    \big(1+\tilde z_0(g,0) \big)
    (-\Delta_{\Z^4})^{-1}_{ab}
    \left( 1 + O\left( \frac{1}{\log|a-b|} \right) \right).
\end{align}
It is a standard fact that $(-\Delta_{\Z^4})^{-1}_{ab} = (2\pi)^{-2}
|a-b|^{-2}(1+O(|a-b|^{-2}))$ (see, e.g., \cite{Lawl91}---the different
constant $(2\pi)^{-2}$ takes into account our definition of the
Laplacian).  Since $\tilde z_0(g,0) = O(g)$, the proof is complete.
(Although our analysis allows $q_j$ to become complex, the left-hand side
of \refeq{final1} is real by definition, so the right-hand side is as well.)
\end{proof}

\begin{rk}
The proof of \refeq{final1} used the fact, proved in
Lemma~\ref{lem:lamlim}, that $\lambda_{x,j}^* \to 1$ as $j \to
\infty$.  The fact that this limit is exactly equal to $1$ (without
$O(g)$ error, as one might expect) is intimately related to the
interpretation of $\lim_{j\to\infty} (1+z_0)^{1/2} \lambda_{x,j}$ as
\emph{field strength renormalisation}.  Without using
$\lambda_{x,\infty}^*=1$, the above proof would show that the
two-point function is asymptotic to
$(1+z_0)\lambda^*_{a,\infty}\lambda^*_{b,\infty}(-\Delta)^{-1}_{ab}$.
Thus, $\lambda_{x,\infty}^*=1$ means that the field strength
renormalisation is given by $(1+z_0)^{1/2}$ only.  This was
anticipated already in
\cite[Section~\ref{log-sec:chvar}]{BBS-saw4-log}, when we split the
original potential $V_{g,\nu,1}$ into an effective free field with
field strength $(1+z_0)^{1/2}$ and mass $m$, and a perturbation.
\end{rk}

\begin{rk}
Note that
\begin{equation}
\lbeq{Gidentity}
    G_{g,\nu_c}(a,b)=(1+\tilde z_0(g,0))q_\infty
\end{equation}
is an \emph{equality}, and not merely an asymptotic
formula.  As such, it contains all information about the two-point
function, including not just the leading asymptotic behaviour but
also all higher-order corrections.
\end{rk}

\begin{rk}
Equations~\refeq{ellsigdef}--\refeq{Kg1} provide corrections
to \cite[(109)--(111)]{BS11}, which contain erroneous powers of $\gbar_N$
in the upper bounds.  In \cite[(109), (111)]{BS11}, the $\gbar_N^3$ in
the upper bound  should be $\gbar_N$, and in \cite[(110))]{BS11} a factor $\gbar_N$
is missing on the right-hand side (it is present in \refeq{ellsigdef}).
The above proof shows that the correct powers here
remain sufficient to prove \refeq{final1}.
\end{rk}

\section*{Acknowledgements}

This work was supported in part by NSERC of Canada.
This material is also based upon work supported by the National
Science Foundation under agreement No.\ DMS-1128155.  RB gratefully
acknowledges the support of the University of British Columbia, where
he was a PhD student while much of his work was done.  Part of this
work was done away from the authors' home institutions, and we
gratefully acknowledge the support and hospitality of the IAM at the
University of Bonn and the Department of Mathematics and Statistics at
McGill University (RB), the Institute for Advanced Study at Princeton
and Eurandom (DB), and the Institut Henri Poincar\'e and the
Mathematical Institute of Leiden University (GS).  We thank Alexandre
Tomberg for many useful discussions, and an anonymous referee for
helpful comments.

\bibliography{../../bibdef/bib}

\begin{thebibliography}{10}

\bibitem{ACG13}
A.~Abdesselam, A.~Chandra, and G.~Guadagni.
\newblock Rigorous quantum field theory functional integrals over the $p$-adics
  {I}: {Anomalous} dimensions.
\newblock Preprint, (2013).

\bibitem{Aize82}
M.~Aizenman.
\newblock Geometric analysis of $\varphi^4$ fields and {I}sing models, {Parts}
  {I} and {II}.
\newblock {\em Commun. Math. Phys.}, {\bf 86}:1--48, (1982).

\bibitem{ACF83}
C.~Arag\~{a}o~de Carvalho, S.~Caracciolo, and J.~Fr\"{o}hlich.
\newblock Polymers and $g|\phi|^4$ theory in four dimensions.
\newblock {\em Nucl. Phys. B}, {\bf 215} [FS7]:209--248, (1983).

\bibitem{Baue13a}
R.~Bauerschmidt.
\newblock A simple method for finite range decomposition of quadratic forms and
  {Gaussian} fields.
\newblock {\em Probab. Theory Related Fields}, {\bf 157}:817--845, (2013).

\bibitem{BBS-saw4-log}
R.~Bauerschmidt, D.C. Brydges, and G.~Slade.
\newblock Logarithmic correction for the susceptibility of the 4-dimensional
  weakly self-avoiding walk: a renormalisation group analysis.
\newblock To appear in \emph{Commun.\ Math.\ Phys.}

\bibitem{BBS-rg-pt}
R.~Bauerschmidt, D.C. Brydges, and G.~Slade.
\newblock A renormalisation group method. {III}. {Perturbative} analysis.
\newblock To appear in \emph{J.\ Stat.\ Phys.}

\bibitem{BBS-phi4-log}
R.~Bauerschmidt, D.C. Brydges, and G.~Slade.
\newblock Scaling limits and critical behaviour of the $4$-dimensional
  $n$-component $|\varphi|^4$ spin model.
\newblock {\em J. Stat. Phys}, {\bf 157}:692--742, (2014).

\bibitem{BDGS12}
R.~Bauerschmidt, H.~Duminil-Copin, J.~Goodman, and G.~Slade.
\newblock Lectures on self-avoiding walks.
\newblock In D.~Ellwood, C.~Newman, V.~Sidoravicius, and W.~Werner, editors,
  {\em Probability and Statistical Physics in Two and More Dimensions}, pages
  395--467. Clay Mathematics Proceedings, vol. 15, Amer. Math. Soc.,
  Providence, RI, (2012).

\bibitem{BFF84}
A.~Bovier, G.~Felder, and J.~Fr\"{o}hlich.
\newblock On the critical properties of the {Edwards} and the self-avoiding
  walk model of polymer chains.
\newblock {\em Nucl. Phys. B}, {\bf 230} [FS10]:119--147, (1984).

\bibitem{BGZ73}
E.~Br\'ezin, J.C. Le~Guillou, and J.~Zinn-Justin.
\newblock Approach to scaling in renormalized perturbation theory.
\newblock {\em Phys.\ Rev.\ D}, {\bf 8}:2418--2430, (1973).

\bibitem{BEI92}
D.~Brydges, S.N. Evans, and J.Z. Imbrie.
\newblock Self-avoiding walk on a hierarchical lattice in four dimensions.
\newblock {\em Ann. Probab.}, {\bf 20}:82--124, (1992).

\bibitem{BS11}
D.~Brydges and G.~Slade.
\newblock Renormalisation group analysis of weakly self-avoiding walk in
  dimensions four and higher.
\newblock In R.~Bhatia et~al, editor, {\em Proceedings of the International
  Congress of Mathematicians, Hyderabad 2010}, pages 2232--2257, Singapore,
  (2011). World Scientific.

\bibitem{BDS12}
D.C. Brydges, A.~Dahlqvist, and G.~Slade.
\newblock The strong interaction limit of continuous-time weakly self-avoiding
  walk.
\newblock In J.-D. Deuschel, B.~Gentz, W.~K\"onig, M.~von Renesse,
  M.~Scheutzow, and U.~Schmock, editors, {\em Probability in Complex Physical
  Systems: In Honour of Erwin Bolthausen and J\"urgen G\"artner}, Springer
  Proceedings in Mathematics, Volume 11, pages 275--287, Berlin, (2012).
  Springer.

\bibitem{BGM04}
D.C. Brydges, G.~Guadagni, and P.K. Mitter.
\newblock Finite range decomposition of {Gaussian} processes.
\newblock {\em J. Stat. Phys.}, {\bf 115}:415--449, (2004).

\bibitem{BI03d}
D.C. Brydges and J.Z. Imbrie.
\newblock {G}reen's function for a hierarchical self-avoiding walk in four
  dimensions.
\newblock {\em Commun. Math. Phys.}, {\bf 239}:549--584, (2003).

\bibitem{BIS09}
D.C. Brydges, J.Z. Imbrie, and G.~Slade.
\newblock Functional integral representations for self-avoiding walk.
\newblock {\em Probab.\ Surveys}, {\bf 6}:34--61, (2009).

\bibitem{BS-rg-loc}
D.C. Brydges and G.~Slade.
\newblock A renormalisation group method. {II}. {Approximation by local
  polynomials}.
\newblock To appear in \emph{J.\ Stat.\ Phys.}

\bibitem{BS-rg-IE}
D.C. Brydges and G.~Slade.
\newblock A renormalisation group method. {IV}. {Stability} analysis.
\newblock To appear in \emph{J.\ Stat.\ Phys.}

\bibitem{BS-rg-step}
D.C. Brydges and G.~Slade.
\newblock A renormalisation group method. {V}. {A} single renormalisation group
  step.
\newblock To appear in \emph{J.\ Stat.\ Phys.}

\bibitem{BS85}
D.C. Brydges and T.~Spencer.
\newblock Self-avoiding walk in 5 or more dimensions.
\newblock {\em Commun. Math. Phys.}, {\bf 97}:125--148, (1985).

\bibitem{CS14}
L.-C. Chen and A.~Sakai.
\newblock Critical two-point functions for long-range statistical-mechanical
  models in high dimensions.
\newblock To appear in \emph{Ann.\ Probab.}

\bibitem{Clis10}
N.~Clisby.
\newblock Accurate estimate of the critical exponent $\nu$ for self-avoiding
  walks via a fast implementation of the pivot algorithm.
\newblock {\em Phys. Rev. Lett.}, {\bf 104}:055702, (2010).

\bibitem{DH92}
J.~Dimock and T.R. Hurd.
\newblock A renormalization group analysis of correlation functions for the
  dipole gas.
\newblock {\em J. Stat. Phys.}, {\bf 66}:1277--1318, (1992).

\bibitem{FMRS87}
J.~Feldman, J.~Magnen, V.~Rivasseau, and R.~S\'en\'eor.
\newblock Construction and {Borel} summability of infrared {$\Phi^4_4$} by a
  phase space expansion.
\newblock {\em Commun. Math. Phys.}, {\bf 109}:437--480, (1987).

\bibitem{Froh82}
J.~Fr\"{o}hlich.
\newblock On the triviality of $\varphi_d^4$ theories and the approach to the
  critical point in $d \geq 4$ dimensions.
\newblock {\em Nucl. Phys.}, {\bf B200} [FS4]:281--296, (1982).

\bibitem{GK85}
K.~Gaw\c{e}dzki and A.~Kupiainen.
\newblock Massless lattice $\varphi^4_4$ theory: Rigorous control of a
  renormalizable asymptotically free model.
\newblock {\em Commun. Math. Phys.}, {\bf 99}:199--252, (1985).

\bibitem{GK86}
K.~Gaw\c{e}dzki and A.~Kupiainen.
\newblock Asymptotic freedom beyond perturbation theory.
\newblock In K.~Osterwalder and R.~Stora, editors, {\em Critical Phenomena,
  Random Systems, Gauge Theories}, Amsterdam, (1986). North-Holland.
\newblock Les Houches 1984.

\bibitem{Hara87}
T.~Hara.
\newblock A rigorous control of logarithmic corrections in four dimensional
  $\varphi^4$ spin systems. {I}. {Trajectory} of effective {Hamiltonians}.
\newblock {\em J. Stat. Phys.}, {\bf 47}:57--98, (1987).

\bibitem{Hara08}
T.~Hara.
\newblock Decay of correlations in nearest-neighbor self-avoiding walk,
  percolation, lattice trees and animals.
\newblock {\em Ann. Probab.}, {\bf 36}:530--593, (2008).

\bibitem{HHS03}
T.~Hara, R.~van~der Hofstad, and G.~Slade.
\newblock Critical two-point functions and the lace expansion for spread-out
  high-dimensional percolation and related models.
\newblock {\em Ann. Probab.}, {\bf 31}:349--408, (2003).

\bibitem{HS92a}
T.~Hara and G.~Slade.
\newblock Self-avoiding walk in five or more dimensions. {I.} {The} critical
  behaviour.
\newblock {\em Commun.\ Math.\ Phys.}, {\bf 147}:101--136, (1992).

\bibitem{HT87}
T.~Hara and H.~Tasaki.
\newblock A rigorous control of logarithmic corrections in four dimensional
  $\varphi^4$ spin systems. {II}. {Critical} behaviour of susceptibility and
  correlation length.
\newblock {\em J. Stat. Phys.}, {\bf 47}:99--121, (1987).

\bibitem{Holl09}
F.~den Hollander.
\newblock {\em Random Polymers}.
\newblock Springer, Berlin, (2009).
\newblock Lecture Notes in Mathematics Vol. 1974. Ecole d'Et\'{e} de
  Probabilit\'{e}s de Saint--Flour XXXVII--2007.

\bibitem{IM94}
D.~Iagolnitzer and J.~Magnen.
\newblock Polymers in a weak random potential in dimension four: rigorous
  renormalization group analysis.
\newblock {\em Commun. Math. Phys.}, {\bf 162}:85--121, (1994).

\bibitem{LK69}
A.I. Larkin and D.E. Khmel'Nitski\u{i}.
\newblock Phase transition in uniaxial ferroelectrics.
\newblock {\em Soviet Physics JETP}, {\bf 29}:1123--1128, (1969).
\newblock {English} translation of {\it Zh.\ Eksp.\ Teor.\ Fiz.} {\bf 56},
  2087--2098, (1969).

\bibitem{Lawl91}
G.F. Lawler.
\newblock {\em Intersections of Random Walks}.
\newblock Birkh\"{a}user, Boston, (1991).

\bibitem{LSW04}
G.F. Lawler, O.~Schramm, and W.~Werner.
\newblock On the scaling limit of planar self-avoiding walk.
\newblock {\em Proc. Symposia Pure Math.}, {\bf 72}:339--364, (2004).

\bibitem{MS93}
N.~Madras and G.~Slade.
\newblock {\em The Self-Avoiding Walk}.
\newblock Birkh{\"a}user, Boston, (1993).

\bibitem{McKa80}
A.J. McKane.
\newblock Reformulation of $n \to 0$ models using anticommuting scalar fields.
\newblock {\em Phys. Lett. A}, {\bf 76}:22--24, (1980).

\bibitem{Nien82}
B.~Nienhuis.
\newblock Exact critical exponents of the ${O}(n)$ models in two dimensions.
\newblock {\em Phys. Rev. Lett.}, {\bf 49}:1062--1065, (1982).

\bibitem{PS80}
G.~Parisi and N.~Sourlas.
\newblock Self-avoiding walk and supersymmetry.
\newblock {\em J. Phys. Lett.}, {\bf 41}:L403--L406, (1980).

\bibitem{SBB11}
R.D. Schram, G.T. Barkema, and R.H. Bisseling.
\newblock Exact enumeration of self-avoiding walks.
\newblock {\em J. Stat. Mech.}, P06019, (2011).

\bibitem{Slad06}
G.~Slade.
\newblock {\em The Lace Expansion and its Applications.}
\newblock Springer, Berlin, (2006).
\newblock Lecture Notes in Mathematics Vol. 1879. Ecole d'Et\'{e} de
  Probabilit\'{e}s de Saint--Flour XXXIV--2004.

\bibitem{ST-phi4}
G.~Slade and A.~Tomberg.
\newblock Critical correlation functions for the $4$-dimensional weakly
  self-avoiding walk and $n$-component $|\varphi|^4$ model.
\newblock Preprint, (2014).

\bibitem{WR73}
F.J. Wegner and E.K. Riedel.
\newblock Logarithmic corrections to the molecular-field behavior of critical
  and tricritical systems.
\newblock {\em Phys. Rev. B}, {\bf 7}:248--256, (1973).

\end{thebibliography}
\bibliographystyle{plain}

\end{document}